\documentclass[runningheads,envcountsame]{llncs}
\usepackage[utf8]{inputenc}
\usepackage{graphicx}
\usepackage{amsfonts}
\usepackage{mathtools}
\usepackage{booktabs}
\usepackage{amssymb}
\usepackage{braket}
\usepackage{csquotes}
\usepackage{tabularx}
\usepackage{enumerate}
\usepackage{hyperref}
\usepackage{xcolor}
\usepackage{xspace}

\usepackage{standalone}
\usepackage{tikz}
\usetikzlibrary{arrows}
\usetikzlibrary{calc}
\usetikzlibrary{positioning}
\usetikzlibrary{shapes}

\usepackage[xcolor]{changebar}
\usepackage[normalem]{ulem}

\DeclarePairedDelimiter{\abs}{\lvert}{\rvert} 
\DeclarePairedDelimiter{\ceil}{\lceil}{\rceil}

\DeclareUnicodeCharacter{211D}{\ensuremath{\mathbb{R}}}
\newcommand{\R}{\ensuremath{\mathbb{R}}\xspace}
\newcommand{\N}{\ensuremath{\mathbb{N}}\xspace}

\newcommand{\ACO}{\ensuremath{\mathrm{AC}^0_\mathbb{R}}\xspace}
\newcommand{\FOArb}{\ensuremath{\mathrm{FO}_\R[\Arb]}\xspace}
\newcommand{\C}{\ensuremath{\mathcal{C}}\xspace}
\newcommand{\Cn}{\ensuremath{\cn_n}}
\newcommand{\bO}{\mathcal{O}}
\newcommand{\struc}[1]{\ensuremath{\mathcal{#1}}}
\newcommand{\feq}{=} 
\newcommand{\fassign}{=}
\newcommand{\fsize}{\textit{tree-shape-size}\xspace}
\newcommand{\fs}{\ensuremath{\textit{tss}}}
\newcommand{\imp}{\to}
\newcommand{\biimp}{\leftrightarrow}
\newcommand{\Rinf}{\ensuremath{\R^\infty}}
\newcommand{\new}{\emph}
\newcommand{\defeq}{\coloneqq}
\newcommand{\mysign}{\textit{sign'}}
\newcommand{\ow}[1][w]{\overline{#1}}
\newcommand{\enc}{\text{enc}}
\newcommand{\sumi}[1][i]{\ensuremath{{\textit{sum}_{#1}}}\xspace}
\newcommand{\sumiq}{\ensuremath{{\textit{sum}_i^q}}\xspace}
\newcommand{\prodi}[1][i]{\ensuremath{{\textit{prod}_{#1}}}\xspace}
\newcommand{\prodiq}{\ensuremath{{\textit{prod}_i^q}}\xspace}
\newcommand{\maxi}[1][i]{\ensuremath{\max_{#1}}}
\newcommand{\cn}{\ensuremath{C}}
\newcommand{\circclass}{\ensuremath{\mathfrak{C}}}
\newcommand{\ptr}{\ensuremath{\mathrm{P}_\R}}
\newcommand{\ltr}{\ensuremath{\mathrm{LT}_\R}}
\newcommand{\uni}[1]{\ensuremath{\mathrm{U}_{#1}}}
\newcommand{\unil}{\uni{\ltr}}
\newcommand{\unip}{\uni{\ptr}}
\newcommand{\unifo}{\uni{\FO}}
\newcommand{\FO}{\ensuremath{\mathrm{FO}_\R}\xspace}
\newcommand{\FOR}{\FO}
\newcommand{\SUM}{\ensuremath{\mathrm{SUM}_\R}\xspace}
\newcommand{\PROD}{\ensuremath{\mathrm{PROD}_\R}\xspace}
\newcommand{\MAX}{\ensuremath{\mathrm{MAX}_\R}\xspace}
\newcommand{\FTIME}{\ensuremath{\mathrm{FTIME}_\R}\xspace}
\newcommand{\Arb}{\ensuremath{\mathrm{Arb}_\R}\xspace}

\newcommand{\tal}{\ensuremath{\tau_\mathrm{ar\_circ}}\xspace}
\newcommand{\struct}{\textit{Struct}}
\newcommand{\structR}{\ensuremath{\struct_\R}}
\newcommand{\ol}[1]{\ensuremath{\overline{#1}}\xspace}
\newcommand{\funi}{\ensuremath{\text{U}_{\FOR}}}
\newcommand{\ttup}{\ensuremath{\tau_{\text{tuple}}}\xspace}

\newcommand{\mine}{\ensuremath{\textnormal{min}}\xspace}
\newcommand{\smine}{\ensuremath{\textnormal{second\_min}}\xspace}
\newcommand{\rk}{\ensuremath{\textit{rank}}}
\newcommand{\papp}[1]{\ensuremath{\left\langle p; #1 \right\rangle}}

\newcommand{\tree}{\ensuremath{\textit{tree}}}
\newcommand{\Vars}{\ensuremath{\mathrm{Vars}}\xspace}

\begin{document}
\title{A Logical Characterization of Constant-Depth Circuits over the Reals\thanks{Supported by DFG VO 630/8-1.}}
%
%
\author{Timon Barlag\orcidID{0000-0001-6139-5219} \and
Heribert Vollmer\orcidID{0000-0002-9292-1960}}
\authorrunning{T. Barlag and H. Vollmer}
%
\institute{Leibniz University, Hanover, Germany\\
\email{\{barlag,vollmer\}@thi.uni-hannover.de}\\
\url{https://www.thi.uni-hannover.de}}
\maketitle              
\begin{abstract}

In this paper we give an Immerman Theorem for real-valued computation, i.e.,  we define circuits of unbounded fan-in operating over real numbers and show that families of such circuits of polynomial size and constant depth decide exactly those sets of vectors of reals that can be defined in first-order logic on \R-structures in the sense of Cucker and Meer.

Our characterization holds both non-uniformly as well as for many natural uniformity conditions.

\keywords{Computation over the reals \and descriptive complexity \and con\-stant-depth circuit families.}
\end{abstract}


\newpage

\section{Introduction}

Computational complexity theory is a branch of theoretical computer science which focuses on the study and classification of problems with regard to their innate difficulty. This is done by dividing these problems into classes, according to the amount of resources necessary to solve them using particular models of computation. One of the most prominent such models is the Turing machine -- a machine operating sequentially on a fixed, finite vocabulary. 

If one wishes to study problems based on their parallel complexity or in the domain of the real numbers, one requires different models of computation. Theoretical models exist both for real-valued sequential and for real-valued parallel computation, going back to the seminal work by Blum, Shub and Smale \cite{DBLP:conf/focs/BlumSS88}, see also \cite{DBLP:books/lib/Blum98}. Their aim was to lay the foundation for a theory of scientific computation, an area going back to Newton, Euler and Gauss, with algorithms over the real numbers. Going even a step further, John von Neumann aimed for a formal logic amenable to mathematical analysis and the continuous concept of the real number.

Unlike Turing machines, machines over $\R$ obtain not an unstructured sequence of bits as input but a vector of real numbers or an (encoding of an) $\R$-structure. The respective parallel model we are going to have a closer look at is a real analogue to the arithmetic circuit (see, e.g., \cite{DBLP:books/daglib/0097931}), which, as its name suggests, resembles electrical circuits in its functioning, however, contrary to these our model operates not on electrical signals, i.e., Boolean values, but real numbers. 

Descriptive complexity is an area of computational complexity theory, which groups decision problems into classes not by bounds on the resources needed for their solution but by considering the syntactic complexity of a logical formalism able to express the problems. Best known is probably Fagin's characterization of the class NP as those problems that can be described by existential second-order formulas of predicate logic \cite{fag74}. Since then, many complexity classes have been characterized logically.  Most important in our context is a characterization obtained by Neil Immerman, equating problems decidable by (families of) Boolean circuits of polynomial size and constant depth consisting of gates of unbounded fan-in, with those describable in first-order logic: 

\begin{theorem}[\cite{DBLP:journals/siamcomp/Immerman89}]
$\mathrm{AC}^0=\mathrm{FO}$. 
\end{theorem}

An important issue in circuit complexity is \textit{uniformity}, i.e., the question if a finite description of an infinite family of circuits exists, and if yes, how complicated it is to obtain it. 
Immerman's Theorem holds both non-uniformly, i.e., under no requirements on the constructability of the circuit family, as well as for many reasonable uniformity conditions \cite{DBLP:journals/jcss/BarringtonIS90}. In the non-uniform case, first-order logic is extended by allowing access to arbitrary numerical predicates, in symbols: non-uniform $\mathrm{AC}^0=\mathrm{FO}[\mathrm{Arb}]$.

The rationale behind the descriptive approach to complexity is the hope to make tools from logic on expressive power of languages available to resource-based complexity and use non-expressibility results to obtain lower bounds on resources such as time, circuit size or depth, etc.


Descriptive complexity seems very pertinent for real-valued computation, since formulas operate directly on structured inputs, which seems quite natural, while computation models generally work on encodings of the input structure, which is an additional abstraction.

In the area of descriptive complexity over the reals,  one usually considers \textit{metafinite structures}, that is finite first-order structures enriched with a set of functions into another possibly infinite structure, in our case the real numbers $\R$.
This study was initiated by Grädel and Meer \cite{DBLP:conf/stoc/GradelM95}, presenting logical characterizations of $\mathrm{P}_\R$ and $\mathrm{NP}_\R$. 
Continuing this line of research, Cucker and Meer obtained a few logical characterizations for bounded fan-in real arithmetic circuit classes \cite{DBLP:journals/jsyml/CuckerM99}, which is what the present paper builds on. Cucker and Meer first proved a characterization of $\mathrm{P}_\R$ using fixed-point logic, and building on this characterized the classes of the NC-hierarchy (bounded fan-in circuits of polynomial size and polylogarithmic depth) restricting the number of updating in the definition of fixed points to a polylogarithmic number. 
They leave out the case of the very low circuit class $\mathrm{AC}^0_\R$, a subclass of $\text{NC}^1_\R$. 
We now expand on their research by making the framework of logics over metafinite structures amenable for the description of \textit{unbounded fan-in circuits}; we are particularly concerned with a real analogue to the class $\ACO$ and show that it corresponds to first-order logic over metafinite structures: 

\begin{theorem}\label{main-thm}
$\ACO{}=\FO$. 
\end{theorem}

Cucker and Meer only note that \enquote{the expressive power of first-order logic is not too big} \cite{DBLP:journals/jsyml/CuckerM99} since it can only describe properties in $\mathrm{NC}^1_\R$. In a sense we close the missing detail in their picture by determining a circuit class corresponding to first-order logic.

The logical characterization of Theorem~\ref{main-thm} holds for arbitrary uniformity conditions based on time-bounded construction of the circuit family, in particular \ptr{}-uniformity and \ltr{}-uniformity.
Extending the framework of Cucker and Meer (who only considered uniform circuits families), we also characterize non-uniform $\ACO$ by first-order logic enhanced with arbitrary numerical predicates.

But the most important case of circuit uniformity in this context is maybe the following: In the Boolean and arithmetic context, it is known that the numerical predicates of addition and multiplication play a special role: If we enhance first-order logic by these, we obtain a logic as powerful as constant-depth circuit families where the uniformity condition itself can be specified in first-order logic: U$_\text{FO}\text{-}\mathrm{AC}^0=\mathrm{FO}[+,\times]$ \cite{DBLP:journals/jcss/BarringtonIS90}. We prove a quite analogous result for real computation: Enhancing first-order logic over the reals with the so called sum and product rules defines exactly those metafinite structures that can be recognized by constant-depth polynomial-size circuit families over the reals where the circuit itself seen as a structure can be specified by first-order logic: 

\begin{theorem}
$\unifo\text{-}\ACO=\FO+\SUM+\PROD$.	
\end{theorem}

This paper is structured as follows: In the next section, we introduce the reader to machines and circuits over $\R$ and the complexity classes they define. We also introduce logics over metafinite structures and prove a couple of auxiliary results concerning useful extensions of $\FO$. Section~\ref{main-nu} proves the correspondence between first-order logic on the one hand side and circuit families of constant-depth and polynomial-size on the other hand, first in the non-uniform case.
The remaining three sections then turn to different uniform versions of this correspondence: Quite an easy case is uniformity defined by polynomial-time, which we present in Section~\ref{main-up}. The proof here is just a quite direct adaptation of the non-uniform case. 
More work is required for logtime-uniformity, since because of the restricted power of the uniformity-machine, a particular numbering for the circuit gates is needed.
 We handle this case in Section~\ref{main-ult}, and we also present a general theorem for uniformity given by arbitrary time-bounds in that section. Finally, in Section~\ref{main-ufo}, we turn to uniformity defined by logical formulas, FO-uniformity. In the technically maybe most-challenging result, we prove the correspondence between \ACO-circuit families with FO-uniformity and \FO enhanced by the addition and multiplication rule, though it should be mentioned that the basic proof structure is still identical to the one for the non-uniform case.
We close by mentioning some questions for further work.


\section{Preliminaries}

In this section, we give an introduction to the machine models and logic over \R{} used in this paper -- which are mostly taken from Cucker and Meer \cite{DBLP:journals/jsyml/CuckerM99} -- and some extensions thereof which we will make use of later on.

\subsection{Machines over \R{}}

Machines over \R{}, which were first introduced by Blum, Shub and Smale \cite{DBLP:conf/focs/BlumSS88}, operate on an unbounded tape of registers containing real numbers. They can evaluate real polynomials and divisions of real polynomials in a single step and branch out by checking if the number contained in a cell is nonnegative. A function $f$ is said to be (\R{}-)\new{computable} if and only if there exists an \R-machine{} $M$, whose input-output-function is exactly $f$. We say that such a machine works in polynomial (logarithmic) time, if the number of steps it takes before halting when given an input $x \in \Rinf{}$ is bounded by a polynomial (logarithmic) function in $\abs{x}$. Here, \Rinf{} denotes arbitrarily long \R-vectors (i.e., $\Rinf = \bigcup_{k \in \N_0}\R^k$) and $\abs{x}$ denotes the length of $x$,  i.e., if $x \in \R^k$ then $\abs{x} = k$.
 
A more formal definition of these machines can be found can also be found in the paper by Cucker and Meer \cite{DBLP:journals/jsyml/CuckerM99}.

\subsection{Arithmetic Circuits over \R}

Arithmetic circuits over \R{} were first introduced by Cucker \cite{DBLP:journals/jc/Cucker92} and are our main model of computation. We will define them in analogy to how they were defined by Cucker and Meer \cite{DBLP:journals/jsyml/CuckerM99}, however in this paper we consider unbounded fan-in.
Also, we disallow division and subtraction gates, since it is not clear, how these operations would be defined for unbounded fan-in.
Since it can be shown that for decision problems, losing (the bounded version of) those gate types does not change computational power within polynomial size, disallowing them does not relativize our results.

\begin{definition}
We define the \textit{sign} function and one variation as follows:\\
\hspace*{5pt}
\begin{minipage}{0.4\textwidth}
\begin{equation*}
\textit{sign}(x) \defeq \begin{cases} 
      1 & x > 0 \\
      0 & x = 0 \\
      -1 & x < 0
   \end{cases}
\end{equation*}
\end{minipage}\hfill
\begin{minipage}{0.4\textwidth}
\begin{equation*}
\mysign(x) \defeq \begin{cases} 
      1 & x \geq 0 \\
      0 & x < 0 \\
   \end{cases}
\end{equation*}
\end{minipage}\hspace{1cm}

\end{definition}
Since the functions $\mysign$ and \textit{sign} can be obtained from one another, given that $\mysign(x) = \textit{sign}(\textit{sign}(x) + 1)$ and $\textit{sign}(x) = \mysign(x) - \mysign(-x)$, we will use both freely whenever we have either one available.

\begin{definition}
\label{def_arith_circ}
An \new{arithmetic circuit} $C$ over \R{} is a directed acyclic graph. Its nodes (also called gates) can be of the following types:\medskip

\noindent\begin{tabularx}{\textwidth}{lX}
\new{Input nodes} & have {indegree} 0 and contain the respective input values of the circuit.\\
\new{Constant nodes} & have {indegree} 0 and are labelled with real numbers.\\
\new{Arithmetic nodes} & can have an arbitrary {indegree} only bounded by the number of nodes in the circuit. They can be labelled with either $+$ or $\times$.\\
\new{Sign nodes} & have {indegree} 1.\\
\new{Output nodes} & have {outdegree} 1 and contain the output values of the circuit after the computation.
\end{tabularx}\medskip

\end{definition}
Nodes cannot be predecessors of the same node more than once, which leads to the {outdegree} of nodes in these arithmetic circuits being bounded by the number of gates in the circuit.\bigskip

\noindent{}In order to later describe arithmetic circuits, we associate with each gate a number which represents its type.
For a gate $g$ these associations are as follows:
\begin{center}
\begin{tabular}{ccccccc}
\toprule
$g$ & input & constant & $+$ & $\times$ & sign & output \\
\midrule
type & 1 & 2 & 3 & 4 & 5 & 6\\
\bottomrule
\end{tabular}
\medskip
\end{center}
For convenience, we define auxiliary gates with types $7$--$12$ which do not grant us additional computational power as we show in Lemma~\ref{lem_aux_gates}. Those are arithmetic gates labelled with $-$ or the relation symbols $=$, $<$, $>$, $\leq$ and $\leq$. All of those nodes have {indegree} 2.

\begin{center}
\begin{tabular}{ccccccc}
\toprule
$g$ & $-$ & $=$& $<$ & $>$ & $\leq$ & $\geq$\\
\midrule
type & 7 & 8 & 9 & 10 & 11 & 12\\
\bottomrule
\end{tabular}
\end{center}
We will also refer to nodes of the types $8-12$ as \new{relational} nodes.


Arithmetic nodes compute the respective function they are labelled with (with $=, <, >, \leq$ and $\geq$ representing their respective binary characteristic functions) and sign gates compute the \textit{sign} function. 
On any input $x$, a circuit $C$ computes a function $f_C$ by evaluating all gates according to their labels. The values of the output gates at the end of the computation are the result of the computation of $C$, i.e., $f_C(x)$.

In order to talk about complexity classes of arithmetic circuits, one considers the \new{depth} and the \new{size} of the circuit. The depth of a circuit is the longest path from an input gate to an output gate and the size of a circuit is the number of gates in a circuit.
\begin{figure}
\begin{center}
\begin{tikzpicture}[
	scale=0.6,
	every node/.style={transform shape}, 
	base/.style={circle,draw,minimum size=30pt}, 
	circ/.style={minimum size=35pt},
	triangle/.style={regular polygon, regular polygon sides = 3, draw, inner sep=0, text width=15mm}
]
\tikzstyle{level 1}=[sibling distance=60mm]
\tikzstyle{level 2}=[sibling distance=40mm]
\tikzstyle{level 3}=[sibling distance=30mm]
\node[base] (out1) {$out$}
	child { node[base] (+1) {$+$} 
		child {node[missing] (emp11) {} edge from parent[draw=none] 
			child {node[base] (lt1) {$\times$} edge from parent[draw=none]
				child {node[base] (c61) {$6$}}
				child {node[base] (i11) {$in_1$}}
			}
		}
		child {node[base] (rt1) {$\times$}
			child{node[missing] (emp21) {} edge from parent[draw=none]
				child{node[base] (i21) {$in_2$} edge from parent[draw=none] }
			}		
		}
	}
	;
\draw [-] (+1) -- (lt1);
\draw [-] (rt1) -- (lt1);
\draw [-] (rt1) -- (i21);
\end{tikzpicture}
\end{center}
\caption{This circuit has size 7, depth 4 and computes the binary function \mbox{$f(x,y)=(6 \cdot x) + (6 \cdot x) \cdot y$}.}
\label{fig_circ_ex}
\end{figure}
\begin{definition}
We say that a directed acyclic graph $C_{sub}$ is a \new{subcircuit} of a circuit $C$, if and only if $C_{sub}$ is weakly connected (i.e. replacing all directed edges with undirected ones in $C_{sub}$ would produce a connected graph), all nodes and edges in $C_{sub}$ are also contained in $C$ and it holds that if there is a path from an input gate to a gate $g$ in $C$, then this path also exists in $C_{sub}$.
For any node $g$ in $C$, we denote by the subcircuit \new{induced} by $g$ that subcircuit $C_{sub,g}$ of $C$, of which $g$ is the top node. We then also say that $g$ is the \new{root node} of $C_{sub,g}$.
\end{definition}


A single circuit can only compute a function with a fixed number of arguments, which is why we call arithmetic circuits a \new{non-uniform} model of computation. In order to talk about arbitrary functions, we need to consider circuit families, i.e., sequences of circuits which contain one circuit for every input length $n \in \N$. The function computed by a circuit family $\C = (\Cn)_{n \in \N}$ is the function computed by the respective circuit, i.e.,
\begin{equation}
f_\C(x) = f_{\cn_{\abs{x}}}(x).
\end{equation}
A circuit family is said to decide a set if and only if it computes the characteristic function of the set. For a function $f \colon \N \to \N$, we say that a circuit family \C{} is of size $f$ (depth $f$), if the size (depth) of $\Cn$ is bounded by $f(n)$.

\begin{definition}
The class \ACO{} is the class of sets decidable by arithmetic circuit families over \R{} of polynomial size and constant depth.
\end{definition}

\begin{lemma}
\label{lem_aux_gates}
For any arithmetic circuit of polynomial size and constant depth which uses gates of the types $1 - 12$, there exists an arithmetic circuit of polynomial size and constant depth computing the same function, which only uses gates of the types $1-6$.
\end{lemma}
\begin{proof}
Let $C$ be an arithmetic circuit with $n$ input gates which uses gates of the types $1-12$, with $size(C) \leq n^q$ and $depth(C)=d$ for $q, d \in \N$. We will construct a circuit $C'$ of polynomial size and constant depth which computes the same function as $C$.
We start out by $C' = C$ and proceed as follows:
First, since we can represent $t_1 \leq t_2$ by 
\begin{equation}
t_1 \leq t_2 \equiv {t_1 < t_2 \lor t_1 = t_2}
\end{equation}
for all $t_1, t_2 \in \R$, we replace every $\leq$ gate in $C'$ by a sign gate, followed by an addition gate, which in turn has a $<$ gate and a $=$ gate as its predecessors. Those two gates then each have the nodes $p_1$ and $p_2$ as their predecessors.
The sign and addition gate at the top represent the $\lor$ in this construction. The overall increase in size is $3$ per $\leq$ gate, which leads to the overall increase in size being polynomial in the worst case. The increase in depth is at worst $2$ per gate on the longest path from an input gate to the output gate, which means that the overall increase in depth is constant. This means that $C'$ still computes the same function as $C$, its size is still polynomial and its depth is still constant in $n$ and $C'$ now does not contain any $\leq$ gates. For $\geq$ gates, we proceed analogously.
We continue similarly for the other cases: Since we can represent $t_1 = t_2$ by 
\begin{equation}
t_1 = t_2 \equiv {\mysign(-(t_1-t_2)^2)}
\end{equation}
for all $t_1, t_2 \in \R$, we replace every $=$ gate in $C'$ with predecessors $p_1$ and $p_2$ by a sign gate at the top, followed by an addition gate which in turn has a constant gate labeled $1$ and another sign gate as its predecessors. This construction represents $\mysign$. That second sign gate then has a subtraction gate as its predecessor, which has a constant node labeled $0$ and a $\times$ gate as its predecessors. The $\times$ gate has two $+$ gates as its predecessors, which in turn each have the same subtraction gate as their predecessor. That subtraction gate then has $p_1$ and $p_2$ as its predecessors. 
 Note here that the $+$ gates here essentially work as identity gates, and we only need them, to have the value of $(p_1 - p_2)$ be multiplied with itself in the $\times$ node. The overall increase in size per $=$ gate in this construction is $9$, which means that the total overhead in size is still polynomial in the worst case. In terms of depth, the increase is at worst $6$ per gate on the longest path from an input gate to the output gate, meaning that the total increase is still constant. After this step, $C'$ computes the same function as $C$, still has polynomial size and constant depth in $n$ and does not contain any $=$ gates.
The construction for $<$ gates with predecessors $p_1$ and $p_2$ works similarly. We make use of $t_1 < t_2$ being representable by 
\begin{equation}
t_1 < t_2 \equiv 1 - \mysign(t_1 - t_2)
\end{equation}
for all $t_1, t_2 \in \R$. We therefore replace every $<$ gate by a subtraction gate with $1$ and a construction for $\mysign$ as above as its predecessors. The $\mysign$ construction then has a subtraction gate as its predecessor, which in turn has the nodes $p_1$ and $p_2$ as its predecessors.
The increase in size per $<$ gate is $6$, leading to a polynomial increase at worst and the increase in depth is at worst $4$ per $<$ gate on the longest path from an input gate to the output, meaning that the overall overhead is constant. This means that $C'$ still has polynomial size and constant depth in $n$, still computes the same function as $C$ and now does not contain any $<$ gates. We proceed analogously for $>$ gates.
For subtraction gates, we proceed similarly, since we can represent $t_1 - t_2$ by 
\begin{equation}
t_1 - t_2 \equiv t_1 + (-1) \times t_2
\end{equation}
for all $t_1, t_2 \in \R$. We replace every subtraction gate with predecessors $p_1$ and $p_2$ by an addition gate with $p_1$ and a multiplication gate as its predecessors, where the multiplication gate has a constant node labeled $-1$ and the node $p_2$ as its predecessors. 
For each gate, this introduces an increase in size of $2$ per subtraction gate, leading to the overall overhead still being polynomial in the worst case, and an increase in depth of $1$ for each gate on the longest path from an input gate to the output gate, which leads to the overall depth still being constant. Therefore, $C'$ still computes the same function as $C$, has polynomial size and constant depth in $n$ and does not contain any subtraction gates.
In total, $C'$ has polynomial size in $n$, constant depth in $n$, only contains gates of the types $1-6$ and computes the same function as $C$. \qed
\end{proof}

%


The circuit families we have just introduced do not have any restrictions on the difficulty of obtaining any individual circuit. For this reason, we also consider so-called \new{uniform} circuit families.

\begin{definition}
\label{def_uniformity}
We say that a circuit family $\C = (\Cn)_{n \in \N}$ is uniform if for each of its circuit the gates are numbered, the predecessors of each gate are ordered and for any given triple of numbers $(n, v_{nr}, p_{idx})$, a corresponding triple 
$(t, p_{nr}, c)$ \label{par_uniformity}
can be computed by an \R-machine{} $M$, where 
\begin{enumerate}[i)]
\item $t$ is the type of the $v_{nr}$th gate $v$ in $\Cn$,
\item $p_{nr}$ is the number of the $p_{idx}$th predecessor of $v$ and
\item $c$ is the value of $v$ if $v$ is a constant gate, the index $i$ if $v$ is the $i$th input gate and $0$ otherwise.
\end{enumerate}

If $v$ has less than $p_{idx}$ predecessors, $M$ returns $(t, 0, 0)$ and if $v_{nr}$ does not encode  a gate in $\Cn$, $M$ returns $(0, 0, 0)$. 

If this computation only takes logarithmic time in $n$, we call \C{} \ltr{}-uniform. If it takes polynomial time in $n$, we call \C{} \ptr{}-uniform.

For a circuit complexity class \circclass, we will by \unil-\circclass{} denote the subclass of \circclass{}, which only contains sets definable by \ltr{}-uniform circuit families. We will use \unip-\circclass{} to analogously denote those sets in \circclass{} definable by \ptr{}-uniform families.
\end{definition}

\subsection{\R-structures and First-order Logic over \R}
\label{sec_logic}

The logics we use to characterize real circuit complexity classes are based on first-order logic with arithmetics. 

\begin{definition}[{\cite[Definition 7]{DBLP:journals/jsyml/CuckerM99}}]
\label{def_R_structure}
Let $L_s$, $L_f$ be finite vocabularies where $L_s$ can contain function and predicate symbols and $L_f$ only contains function symbols. An \new{\R -structure of signature $\sigma = (L_s, L_f)$} is a pair $\struc{D} = (\struc{A}, \struc{F})$ where 
\begin{enumerate}
\item \struc{A} is a finite structure of vocabulary $L_s$ which we call the \new{skeleton} of \struc{D} whose universe $A$ we will refer to as the \new{universe} of \struc{D} and whose cardinality we will refer to by $\abs{A}$ 
\item and \struc{F} is a finite set which contains functions of the form $X \colon A^k \rightarrow \R$ for $k \in \N$ which interpret the function symbols in $L_f$.
\end{enumerate}

We will use \new{$\structR(\sigma)$} to refer to the set of all \R-structures of signature $\sigma$ and we will assume that for any fixed signature $\sigma = (L_s, L_f)$, we can fix an ordering on the symbols in $L_s$ and $L_f$.
\end{definition}



In order to use \R-structures as inputs for machines, we encode them in $\R^\infty$ as follows: We start by choosing an arbitrary ranking $r$ on $A$, i.e., a bijection ${r \colon A \to \{0, ..., \abs{A}-1\}}$. We then replace all predicates in $L_s$ by their respective characteristic functions and all functions $f \in L_s$ by $r \circ f$. Those functions are then considered to be elements of $L_f$. We represent each of these functions by concatenating their function values in lexicographical ordering on the respective function arguments according to $r$. To encode \struc{D} we only need to concatenate all representations of functions in $L_f$ in the order fixed on the signature. We denote this encoding by \enc(\struc{D}).

In order to be able to compute $\abs{A}$ from \enc(\struc{D}), we make an exception for functions and predicates of arity $0$. We treat those as if they had arity $1$, meaning that e.g. we encode a function $f_1() = 3$ as $\abs{A}$ many $3$s.

\noindent{}Since
\begin{equation}
\abs{\enc(\struc{D})} = \sum\limits_{f \in L_f}\abs{A}^{\max\{ar(f), 1\}},
\end{equation}
\label{par_enc_rec}
where $ar(f)$ is the arity of $f$, we can reconstruct $\abs{A}$ from the arities of the functions in $L_f$ and the length of $\enc(\struc{D})$. We can do so by using for example binary search, since we know that $\abs{A}$ is between $0$ and $\abs{\enc(\struc{D})}$. We can therefore compute $\abs{A}$ when given $\varphi$ and $\abs{\enc(\struc{D})}$ in time logarithmic in $\abs{\enc(\struc{D})}$.

\subsubsection{First-order Logic over \R}\leavevmode

\begin{definition}[First-order logic]
The language of first-order logic contains for each signature $\sigma = (L_s, L_f)$ a set of formulas and terms. The terms are divided into \new{index terms} which take values in universe of the skeleton and \new{number terms} which take values in \R. These terms are inductively defined as follows:
\begin{enumerate}
\item The set of index terms is defined as the closure of the set of variables $\Vars$ under applications of the function symbols of $L_s$.
\item Any real number is a number term.
\item For index terms $h_1, ..., h_k$ and a $k$-ary function symbol $X \in L_f$, $X(h_1, ..., h_k)$ is a number term.
\item If $t_1$, $t_2$ are number terms, then so are $t_1 + t_2$, $t_1 \times t_2$ and $\textit{sign}(t_1)$.
\end{enumerate}

Atomic formulas are equalities of index terms $h_1 \feq h_2$ and number terms $t_1 \feq t_2$, inequalities of number terms $t_1 < t_2$ and expressions of the form $P(h_1, ..., h_k)$, where $P \in L_s$ is a k-ary predicate symbol and $h_1, .., h_k$ are index terms.

The set \FO{} is the smallest set which contains the closure of atomic formulas under the Boolean connectives $\{\land, \lor, \neg, \imp, \biimp\}$ and quantification $\exists v \psi$ and $\forall v \psi$ where $v$ ranges over \struc{A}. 
\end{definition}

Equivalence of \FO{} formulas and sets defined by \FO{} formulas are done in the usual way, i.e., a formula $\varphi$ defines a set $S$ if and only if the elements of $S$ are exactly the encodings of \R-structures under which $\varphi$ holds and two such formulas are said to be equivalent if and only if they define the same set.

\subsubsection{Extensions to \texorpdfstring{\FO{}}{FO}}\leavevmode

In the following, we would like to extend \FO{} by additional functions and relations that are not given in the input structure. To that end, we make a small addition to Definition~\ref{def_R_structure} where we defined \R-structures. Whenever we talk about \R-structures over a signature $(L_s, L_f)$, we now also consider structures over signatures of the form $(L_s, L_f, L_a)$. The additional vocabulary $L_a$ does not have any effect on the \R-structure, but it contains function and relation symbols, which can be used in a logical formula with this signature. This means that any \R-structure of signature $(L_s, L_f)$ is also an \R-structure of signature $(L_s, L_f, L_a)$  for any vocabulary $L_a$.

\begin{definition}


Let $R$ be a set of finite relations and functions. We will write \FO$[R]$ to denote the class of sets that can be defined by \FO-sentences which can make use of the functions and relations in $R$ in addition to what they are given in their structure. 
Formally, this means that \FO$[R]$ describes exactly those sets $S \subseteq \Rinf$ for which 
	there exists an \FO{}-sentence $\varphi$ over a signature ${\sigma = (L_s, L_f, L_a)}$ such that for each length $n$,
		there is an interpretation $I_n$ interpreting the symbols in $L_a$ as elements of $R$ such that 
				for all \Rinf-tuples $s$ of length $n$ it holds that
					$s \in S$ if and only if $s$ encodes an \R-structure over $(L_s, L_f, L_a)$ which models $\varphi$ when using $I_n$.

\end{definition}

With the goal in mind to create a logic which can define sets decided by circuits with unbounded {fan-in}, we introduce new rules for building number terms: the \textit{sum} and the \textit{product rule}. We will also give another rule, which we call the \textit{maximization rule}, but will later show that we can define this rule in \FO{} and thus do not gain expressive power by using it. We will use this rule to show that we can represent characteristic functions in \FO{}.

\begin{definition}[sum, product and maximization rule]
Let $t$ be a number term in which the variable $i$ occurs freely with other variables $\ow = w_1, ..., w_j$ and let $A$ denote the universe of the given input structure. Then 
\begin{equation}
\sumi(t(i, \ow))
\end{equation}
is also a number term which is interpreted as $\sum_{i \in A}t(i, \ow)$.
The number terms $\prodi(t(i, \ow))$ and $\maxi(t(i, \ow))$ are defined analogously.

We also write $\sumiq(t(i_1, ..., i_q, \ow))$ to denote $\sumi[i_1](... \sumi[i_q](t(i_1, ..., i_q, \ow)))$ for convenience and we will use \prodiq{} analogously.
\end{definition}

For a logic $\mathcal{L}$, we will by $\mathcal{L} + \SUM$, $\mathcal{L} + \PROD$ and $\mathcal{L} + \MAX$ denote $\mathcal{L}$ extended by the sum rule, the product rule or the maximization rule respectively.

We will now evaluate which logics can already natively use some of the aforementioned rules. As it turns out, the maximization rule can be used in \FO{} without any extensions and the sum and product rule extend neither \FOArb{} nor a polynomial extension of \FO{} which we will see later.

\begin{lemma}
\label{lem_fo_max}
$\FO = \FO + \MAX$
\end{lemma}
\begin{proof}
For each $\FO + \MAX$ formula, we can construct an equivalent \FO formula. For each such term containing $\maxi(F(i))$, the basic idea it to add a quantifier prefix which makes sure that there exists an element $x \in A$ such that for all elements $y \in A$, $F(x) \geq F(y)$.

Let $\varphi$ be a \FO{} formula which contains $\maxi$-constructions, i.e., number terms of the form $\maxi(t(i, \ow))$ for a number term $t$. We will show that for every such formula, we can construct another \FO formula $\varphi'$ which is equivalent to $\varphi$ but which does not contain the term $\maxi(t(i, \ow))$. Since $\maxi$-constructions are number terms, whenever they occur, they are part of atomic (sub-)formulas. For this reason, we only need to show, how to turn atomic formulas with $\maxi$-constructions into semantically equivalent formulas (that are not necessarily atomic anymore).
For a given atomic formula with $\maxi$-constructions $\varphi$, define $\varphi'$ as follows:
\noindent Let $\varphi \fassign t_1 \feq t_2$ and let $\maxi[i_1], ..., \maxi[i_k]$ be the $\maxi$-occurrences of $\varphi$, ordered by level of nesting, where $\maxi[i_1]$ has the lowest level of nesting, the nesting of $\maxi[i_2]$ is either the same as $\maxi[i_1]$ or greater by $1$ and so on. We assume without loss of generality that the variables $x_1, ..., x_k$ and $y_1, ..., y_k$ do not occur in $\varphi$. We also assume for now that there is only one occurrence of $\maxi$ at the lowest level of nesting and that $t_1$ consists only of that outermost $\maxi$-construction, i.e., $t_1 \fassign \maxi[i_1](F_1(i_1, \ow_1))$. To now construct $\varphi'$, we go through the $\maxi$-occurrences in $\varphi$ in reverse order of nesting, i.e., from the deepest level to the shallowest, and for each occurrence $\maxi[i_m](F_m(i_m, \ow_m))$, we create a subformula $\psi_m$, which ensures that $F_m$ is being maximized with respect to $i_m$. We will use new variables $x_1, ..., x_k, y_1, ..., y_k$ in the subformulas, which will be quantified later, when we connect those subformulas to construct $\varphi'$. $\varphi'$ will then have the form 
\begin{equation}
\varphi' \fassign \exists x_1 \forall y_1 ... \exists x_k \forall x_k \psi_k \land ... \land \psi_1 \land \widehat{\varphi},
\end{equation}
where $\widehat{\varphi}$ represents the structure of $\varphi$ without any $\maxi$-constructions. In our case, $\widehat{\varphi}$ would just be $F_1(x_1, \ow_1) \feq t_2$.
	
We start with the term $\maxi[i_k](F_k(i_k, i_{k_1}, ..., i_{k_j}, \ow_k))$, where $F_k$ is the number term in $\varphi$ getting maximized by $\maxi[i_k]$, $i_{k_1}, ..., i_{k_j}$ are the variables used in $F_k$ from $\maxi$-constructions which occur at lower levels of nesting in $\varphi$ and $\ow_k$ are all other variables used in $F_k$. 
	
We now create the subformula 
\begin{equation}
\psi_{i_k} \fassign F_k(x_k, x_{k_1}, ..., x_{k_j}, \ow_k) \geq F_k(y_k, x_{k_1}, ..., x_{k_j}, \ow_k),
\end{equation}
	
\noindent{}which makes sure that $F_k$ is maximal with respect to $i_k$.
	
Afterwards, we proceed in reverse order of nesting with the other $\maxi$-occurrences in $\varphi$ (meaning that $\maxi[i_{k-1}]$ is next) and create the subformulas $\psi_{k-1}, ..., \psi_1$ similarly. For $m \in (k-1, ..., 1)$, we proceed as follows:
	
Let $\maxi[i_m](F_m(i_m, i_{m_1}, ..., i_{m_j}, \ow_m))$ be the occurrence of $\maxi[i_m]$ in $\varphi$ with analogous $F_m, i_m, i_{m_1}, ..., i_{m_j}, \ow_m$ as before. Now replace all $\maxi$-constructions $\maxi(F_i(i, \ow))$ in $F_m$ -- where $\ow$ are all variables used in $F_i$ except for $i$ -- by parentheses around $F_i$, i.e., $\maxi(F_i(i, \ow))$ would just become $(F_i(i, \ow))$. Denote the result by $F_m'$. We then define 
\begin{equation}
\psi_m \fassign F_m'(x_m, x_{m_1}, ..., x_{m_j}, \ow_m) \geq F_m'(y_m, x_{m_1}, ..., x_{m_j}, \ow_m).
\end{equation}	
Finally, we define
\begin{equation}
\label{lem_maxi_phi'_full} 
\varphi' \fassign \exists x_1 \forall y_1 ... \exists x_k \forall y_k ~ \psi_k \land ... \land \psi_1 \land F_1'(x_1, \ow_1) \feq t_2.
\end{equation}
This construction now works for our strong assumption that $t_1 \fassign \maxi[i_1](F_{1})$. However, we only require the following modifications to make it generally applicable: If $\varphi$ contains only one $\maxi$-construction at the lowest level, but then operates on that construction, we can just add the context of that $\maxi$-construction to the term $F_1'$ in $\varphi'$. For example if $\varphi \fassign 7 \feq \maxi(F(i)) + 1$, then we could just add the '$+1$' to the $F_1'(x_1, \ow_1)$ in Formula~\ref{lem_maxi_phi'_full}. If $\varphi$ contains several $\maxi$-constructions at the lowest level of nesting, then we can construct as we have previously and just add the subformulae to the conjunction in $\varphi'$.
	
$\varphi'$ now does not contain any $\maxi$-constructions and is therefore a valid \FO{} formula. Since for every $\maxi$-occurrence in $\varphi$, there is a subformula in the conjunction of $\varphi'$ making sure that the term maximized by $\maxi$ in $\varphi$ is also maximal in $\varphi'$, $\varphi'$ is also semantically equivalent to $\varphi$.
	
We can construct $\varphi'$ analogously, if both, $t_1$ and $t_2$ contain $\maxi$-constructions or if $\varphi \fassign t_1 < t_2$. We have therefore shown that for any \FO{} formula with $\maxi$-constructions, there exists a semantically equivalent formula which does not contain any such constructions. \qed

\end{proof}



\begin{remark}
For the sake of simplicity we only consider \new{functional \R-structures} in the following, i.e., \R-structures whose signatures do not contain any predicate symbols. This does not restrict what we can express, since any relation $P \in A^k$ can be replaced by its characteristic function $\chi_P \colon A^k \to \{0,1\}$.
\end{remark}

As mentioned before, the reason why we need the maximization rule is that we would like to write characteristic functions as number terms. 
This will become useful when we characterize our circuit models logically.
For a first-order formula $\varphi(v_1,...,v_r)$ we define its characteristic function $\chi[\varphi]$ on a structure \struc{D} by
\begin{equation}
\chi[\varphi](a_1, ..., a_r) = \begin{cases}
1 & \text{if } \struc{D} \models \varphi(a_1, ..., a_r)\\
0 & \text{otherwise}
\end{cases}
\end{equation}

The following result is a slight modification of a result presented by Cucker and Meer \cite{DBLP:journals/jsyml/CuckerM99}. 

\begin{proposition}[\cite{DBLP:journals/jsyml/CuckerM99}]\label{prop_char_func}
Let $R$ be a set of functions and predicates. For every $\FO[R]$ formula $\varphi$, there is an $\FO[R]$ number term $t_{\chi[\varphi]}$ such that for all structures $\struc{D}$ it holds that $t_{\chi[\varphi]}$, when occurring in another formula, evaluates to $1$ under $\struc{D}$ if $\struc{D} \models \varphi$ and to $0$, otherwise. 
\end{proposition}
\begin{proof}
We will prove this proposition by induction on the construction of $\varphi$. If $\varphi$ is atomic, then it is of the form  $t_1 \feq t_2$, $t_1 < t_2$ for number terms $t_1, t_2$, since we only consider functional \R-structures. For atomic formulas, we have
\begin{equation}
\chi[t_1 \feq t_2] = \mysign[-(t_1 - t_2)^2]
\end{equation}
and
\begin{equation}
\chi[t_1 < t_2] = 1 - [\mysign(t_1 - t_2)].
\end{equation}
If $\varphi$ is of the form $\varphi = \exists x \psi(x)$, then
\begin{equation}
\chi[\varphi] = \maxi[x] \chi[\psi(x)].
\end{equation}
If $\varphi$ has the form $\varphi = \neg \psi$, then
\begin{equation}
\chi[\varphi] = 1 - \chi[\psi]
\end{equation}
and if $\varphi = \psi \land \xi$, then
\begin{equation}
\chi[\varphi] = \chi[\psi] \times \chi[\xi].
\end{equation}
Since $\varphi = \forall x \psi(x)$ and the remaining Boolean connectives can be constructed from the above, we have now shown that we can describe $\chi[\varphi]$ in $\FOR[R]$ for any $\varphi \in \FOR[R]$. 
\qed
\end{proof}

We will write $\chi[\varphi]$ to denote the use of $t_{\chi[\varphi]}$ when writing number terms.

\begin{remark}
The restriction in Proposition~\ref{prop_char_func} that we can only define those number terms when occurring in $\FO$ formulas stems from the way we showed that $\FO + \MAX = \FO$ in Lemma~\ref{lem_fo_max} and is for our intents and purposes the most part negligible.
However, later on, when talking about $\FO$-uniformity, we would like to be able to define number terms even occurring outside of formulas.
(Particularly for specifying $\varphi_\mathrm{const\_val}$ in the proof of Theorem~\ref{thm_UFO_AC0})
We therefore add the following corollary.
\end{remark}

\begin{corollary}
Let $R$ be a set of functions and predicates. For every $\FO[R]+\SUM+\PROD$ formula $\varphi$, there is an $\FO[R]+\SUM+\PROD$ number term $t_{\chi[\varphi]}$ such that for all structures $\struc{D}$ it holds that $t_{\chi[\varphi]}$ evaluates to $1$ under $\struc{D}$ if $\struc{D} \models \varphi$ and to $0$, otherwise. 
\end{corollary}
\begin{proof}
This proof works identically to the proof for Proposition~\ref{prop_char_func} except for the case $\chi[\exists x \psi(x)]$.
In that case, we now define
\begin{equation}
\chi[\varphi] = \textit{sign}(\sumi[x](\chi[\psi(x)])).
\end{equation}\qed
\end{proof}

\subsection{Logical Uniformity}

Since circuits and logic seem to be closely related, there is another type of uniformity, different to the one defined in Definition~\ref{def_uniformity}, that we would like to have a look at.
First-order uniform circuit complexity classes are those classes, of which the respective circuit families can be described by first-order formulas and terms. 
This kind of uniformity was introduced to Boolean circuit complexity by \cite{DBLP:journals/jcss/BarringtonIS90} and we are now going to map that concept to our circuit classes over the reals.

\begin{definition}
Let the vocabulary of arithmetic circuits \tal{} be defined as follows:
\[
\tal \coloneqq ((\varphi_+^1, \varphi_\times^1, \varphi_{\mathrm{sign}}^1, \varphi_{\mathrm{input}}^2, \varphi_\mathrm{E}^2, \varphi_{\mathrm{output}}^1, \varphi_{\mathrm{const}}^1 ), ( \varphi_{\mathrm{const\_val}}^1 ))
\]

%
\end{definition}

\begin{definition}
Let the vocabulary of \R-tuples \ttup be defined as follows:
\[
\ttup \coloneqq (( \leq^2 ), ( f_{\text{element}}^1 ))
\]
\end{definition}

\begin{definition}
Let $\sigma_s, \tau_s$ be vocabularies of relation symbols and function symbols and let $\sigma_f, \tau_f$ vocabularies of function symbols.
Additionally, let $\tau_s = (g_1^{ar(g_1)}, \ldots, g_p^{ar(g_p)}, h_1^{ar(h_1)}, \ldots, h_q^{ar(h_q)})$ and $\tau_f = (f_1^{ar(f_1)}, \ldots, f_r^{ar(f_r)})$, where $ar(f)$ denotes the arity of $f$ for all function and relation symbols $f$ and let $k \in \N$.

A \emph{real first-order interpretation} (\emph{\FOR-interpretation})
\[
I \colon \structR[(\sigma_s, \sigma_f)] \to \structR[(\tau_s, \tau_f)]
\]
is given by a tuple of \FOR{} formulae $\varphi_0, \varphi_1 \ldots, \varphi_p$, \FOR{} index terms $a_1, \ldots, a_q$ and \FOR{} number terms $t_1, \ldots, t_r$ over $(\sigma_s, \sigma_f)$. 
The formula $\varphi_0$ has $k$ free variables, $\varphi_i$ has $k \cdot ar(g_i)$ free variables for $1 \leq i \leq p$, $a_i$ has $k \cdot ar(h_i)$ free variables for $1 \leq i \leq q$ and $t_i$ has $k \cdot ar(f_i)$ free variables for $1 \leq i \leq r$.

For each structure $\mathcal{A} \in \structR[(\sigma_s, \sigma_f)]$, these terms and formulae define the structure
\begin{align*}
I(\mathcal{A}) &= ~(\abs{I(\mathcal{A})}, g_1^{I(\mathcal{A})}, \ldots, g_p^{I(\mathcal{A})}, h_1^{I(\mathcal{A})}, \ldots, h_p^{I(\mathcal{A})}, f_1^{I(\mathcal{A})}, \ldots, f_r^{I(\mathcal{A})}) \\
& \in \structR[(\tau_s, \tau_f)],
\end{align*}
where the universe is defined by $\varphi_0$ and the functions and relations are defined by $\varphi_1, \ldots, \varphi_p, a_1, \ldots, a_q, t_1, \ldots, t_r$ in the following way:
\begin{align*}
\abs{I(\mathcal{A})} =~ & \{(b^1, \ldots, b^k) \mid \mathcal{A} \models \varphi_0(b^1, \ldots, b^k)\}\\
g_i^{I(\mathcal{A})} =~ & \{(\ol{b_1}, \ldots, \ol{b_{ar(g_i)}}) \in \abs{I(\mathcal{A})}^{ar(g_i)} \mid \mathcal{A} \models \varphi_i(\ol{b_1}, \ldots, \ol{b_{ar(g_i)}}) \}\\
h_i^{I(\mathcal{A})}(\ol{b_1}, \ldots, \ol{b_{ar(h_i)}}) =~ & \ol{b}, \text{ iff } a_i(\ol{b_1}, \ldots, \ol{b_{ar(h_i)}}) \models_\mathcal{A} \ol{b}\\
f_i^{I(\mathcal{A})}(\ol{b_1}, \ldots, \ol{b_{ar(f_i)}}) =~ & \ol{b}, \text{ iff } t_i(\ol{b_1}, \ldots, \ol{b_{ar(f_i)}}) \models_\mathcal{A} \ol{b},
\end{align*}
where $a(x_1, \dots x_{ar(a)}) \models_\mathcal{A} b$ means that the term $a$ when interpreted with the structure $\mathcal{A}$ evaluates to $b$.

\end{definition}

\begin{definition}
An \R-circuit family $\mathcal{C}=(C_n)_{n \in \N}$ is said to be \FOR-uniform if there is an \FOR-interpretation 
\[
I \colon \structR[\ttup] \to \structR[\tal]
\]
mapping any $R$-structure $\mathcal{A}$ over $\ttup$ to the circuit $C_{\abs{\enc(\mathcal{A})}}$ given as a structure of vocabulary \tal{}. 
This means that the symbols of $\ttup$ are interpreted as follows:

\begin{itemize}
\item $x \leq y$: $x$ is ranked lower than $y$ in $\mathcal{A}$
\item $f_\mathrm{element}(x) = y \in \R$: the $\rk(x)$th value in the encoding of $\mathcal{A}$ is $y$
\end{itemize}

and the symbols of $\tal$ are interpreted in the following way:

\begin{itemize}
\item $\varphi_0(\ol{x})$: $\ol{x}$ is a gate.	
\item $\varphi_+(\ol{x})$: $\ol{x}$ is an addition gate.
\item $\varphi_\times(\ol{x})$: $\ol{x}$ is a multiplication gate.
\item $\varphi_\mathrm{sign}(\ol{x})$: $\ol{x}$ is a sign gate.
\item $\varphi_+(\ol{x}, \ol{i})$: $\ol{i} = (j, \dots, j)$ and $\ol{x}$ is the $\rk(j)$th input gate, where $\rk$ is the ranking of $\mathcal{A}$.
\item $\varphi_E(\ol{x}, \ol{y})$: $\ol{y}$ is a successor gate of $\ol{x}$.
\item $\varphi_\mathrm{output}(\ol{x})$: $\ol{x}$ is the output gate.
\item $\varphi_\mathrm{const}(\ol{x})$: $\ol{x}$ is a constant gate.
\item $\varphi_\mathrm{const\_val}(\ol{x}) = y \in \R$: $y$ is the value of $\ol{x}$ if $\ol{x}$ is a constant gate and $y = 0$,
\item[] otherwise.
\end{itemize}
\end{definition}

\begin{definition}
Let $\mathfrak{C}$ be a complexity class defined by a non-uniform circuit families over \R.
Then $\funi{}$-$\mathfrak{C}$ consists of all languages in $\mathfrak{C}$ which are defined by $\FOR${}-uniform \R-circuit families.
\end{definition}


\section{A Characterization for Non-Uniform \texorpdfstring{\ACO}{AC0}}
\label{main-nu}

In the upcoming sections, we give descriptive complexity results for the non-uniform set \ACO{} and some of its uniform subsets. 
In order to achieve this, we use the previously defined first-order logic over the real numbers and the extensions we defined. 

First of all we show an equality which is close to a classical result shown by Immermann \cite{Immerman87languagesthat}. 
We show that extending our first-order logic over the reals with arbitrary functions lets us exactly describe the non-uniform set \ACO{}.\medskip

In the proof for the upcoming theorem, we make use of a convenient property of circuits deciding \ACO{}-sets, namely that for each of those circuits, there exist \textit{tree-like} circuits deciding the same set.  
We call a circuit tree-like, if it is a directed tree with the exception of the input nodes.
Those nodes, which would represent the leaves, can have multiple successor nodes. 
That means that tree-like circuits are trees up until the penultimate level and would be actual trees, if one would copy every input gate for each outgoing edge, rather than letting them have multiple successors.

\begin{lemma}
\label{lem_treelike}
For every \ACO-circuit family $(\Cn)_{n \in \N}$, there exists a tree-like \ACO-circuit family $(\cn'_n)_{n \in \N}$ computing the same function, such that for all $n \in \N$ and every gate $v$ in $\cn'_n$, every path from an input gate to $v$ has the same length.
\end{lemma}
\begin{proof}
In order to prove this we show that any \ACO-family can be transformed into an \ACO-family which exhibits the specified property. 
For any given circuit of a \ACO-family, we first make sure that all non-input gates have outdegree $1$. 
In order to achieve this, for each gate $g$ with outdegree $k > 1$ we copy the subcircuit $C_{sub,g}$ induced by $g$ $k - 1$ times, so that we now have $k$ copies of $C_{sub,g}$. 
For each of the previously outgoing edges $g \to v$ of $g$, the root of one of the copies of $C_{sub,g}$ then has $v$ as its (sole) successor.

We do this iteratively, in each step only modifying gates with outdegree $\geq 2$ that are closest to input gates. 
Afterwards, we pad all paths from input gates to the output gate with addition gates to ensure that they have the same length. 
This can be done with only a polynomial overhead in size and a constant overhead in depth without changing the computed function. 
Figure~\ref{fig_treelike_ex} shows an example of this construction.
\qed
\end{proof}

\begin{figure}
\begin{center}
\resizebox{\textwidth}{!}{\begin{tikzpicture}[
	scale=0.6,
	every node/.style={transform shape}, 
	base/.style={circle,draw,minimum size=30pt}, 
	circ/.style={minimum size=35pt},
	triangle/.style={regular polygon, regular polygon sides = 3, draw, inner sep=0, text width=15mm}
]
\tikzstyle{level 1}=[sibling distance=60mm]
\tikzstyle{level 2}=[sibling distance=40mm]
\tikzstyle{level 3}=[sibling distance=30mm]
\node[base] (out1) {$out$}
	child { node[base] (+1) {$+$} 
		child {node[missing] (emp11) {} edge from parent[draw=none] 
			child {node[base] (lt1) {$\times$} edge from parent[draw=none]
				child {node[base] (c61) {$6$}}
				child {node[base] (i11) {$in_1$}}
			}
		}
		child {node[base] (rt1) {$\times$}
			child{node[missing] (emp21) {} edge from parent[draw=none]
				child{node[base] (i21) {$in_2$} edge from parent[draw=none] }
			}		
		}
	}
	;
\draw [-] (+1) -- (lt1);
\draw [-] (rt1) -- (lt1);
\draw [-] (rt1) -- (i21);
\tikzstyle{level 3}=[sibling distance=25mm]
\node[base, right=7cm of out1] (out2) {$out$}
	child { node[base] (+2) {$+$} 
		child {node[missing] (emp12) {} edge from parent[draw=none] 
			child {node[base] (lt2) {$\times$} edge from parent[draw=none]
				child {node[base] (c62) {$6$}}		
				child {node[base] (ec62) {$6$} edge from parent[draw=none]}		
			}
			child {node[base] (elt2) {$\times$} edge from parent[draw=none]
				child[missing]
				child {node[base] (i12) {$in_1$}}
			}
		}
		child {node[base] (rt2) {$\times$}
			child{node[missing] (emp22) {} edge from parent[draw=none]
				child{node[base] (i22) {$in_2$} edge from parent[draw=none] }
			}		
		}
	}
	;
\draw [-] (+2) -- (lt2);
\draw [-] (rt2) -- (elt2);
\draw [-] (rt2) -- (i22);
\draw [-] (i12) -- (lt2);
\draw [-] (ec62) -- (elt2);
\node[base, right=6cm of out2] (out3) {$out$}
	child { node[base] (+3) {$+$} 
		child {node[base] (lt3) {$\times$}
			child {node[base] (c63) {$6$}}
			child {node[base] (i+1) {$+$} 
				child {node[missing] (i13) {} edge from parent[draw=none]}				
			}
		}
		child {node[base] (rt3) {$\times$}
			child {node[base] (elt3) {$\times$} 
				child[missing]
				child {node[base] (in1) {$in_1$}}
			}
			child{node[base] (i+2) {$+$}
				child{node[base] (i23) {$in_2$} }
			}		
		}
	}
	;
\node[base] (ec63) at ($(i+1) !.5! (elt3) + (i13) - (i+1)$) {$6$};
\draw [-] (ec63) -- (elt3);
\draw [-] (in1) -- (i+1);

\draw [->] ($(+1) + (2.5,0)$) -- ($(+1) + (4,0)$) node[midway, yshift=8pt] {Step 1};
\draw [->] ($(+2) + (3,0)$) -- ($(+2) + (4.5,0)$) node[midway, yshift=8pt] {Step 2};
\end{tikzpicture}}
\end{center}
\caption{An example of turning the circuit from Figure~\ref{fig_circ_ex} into a tree-like circuit as described in Lemma~\ref{lem_treelike}}
\label{fig_treelike_ex}
\end{figure}

Additionally, we would also like to take advantage of a similarly convenient property of the formulas of our real first-order logic.
Function and relation symbols in such formulas can have arbitrary index terms as their arguments, however, it can be shown that for all real first-order formulas, there is an equivalent formula in which all function and relation symbols only have variables in their arguments.
This will be useful when constructing circuits for given formulas in the upcoming proofs.

\begin{lemma}\label{lem_normalform}
For every function $f \colon \N \rightarrow \N$ and every $\FO[\FTIME(f(n))]+\SUM+\PROD$ formula $\varphi$, there is a $\FO[\FTIME(f(n))]+\SUM+\PROD$ formula $\varphi'$ which is equivalent to $\varphi$ but where all function and relation symbols only have variables as their arguments.
\end{lemma}
\begin{proof}
Let $f \colon \N \rightarrow \N$ be a function and $\varphi$ be a $\FO[\FTIME(f(n))]+\SUM+\PROD$ formula in which there are occurrences of relation symbols, of which some arguments are non-variable index terms.
Then $R$ be such a relation symbol in $\varphi$.
Then $\varphi$ contains a subformula $\psi$ of the form 
\[
\psi = R(h_1, \dots, h_k),
\]
where $h_1, \dots, h_k$ are index terms (which might in turn be function symbols applied to more index terms).
There must be a deepest level of nesting, at which there are index terms which only have variables as their arguments.
Let $S$ be a relation symbol which has only this type of function symbol as its arguments. 
(And possibly variables, but those can be ignored here, since we can simply leave them unchanged.)
This means that the occurrence of $S$ has the form 
\[
S(f_1(x_{11}, \dots, x_{1k_1}), \dots, f_\ell(x_{\ell 1}, \dots, x_{\ell k_\ell}))
\]
where $f_1, \dots, f_\ell$ are function symbols. 
Now for each of the function symbols $f_i$, introduce a new variable symbol $y_i$ and proceed by replacing the occurrence $S(f_1(x_{11}, \dots, x_{1k_1}), \dots, f_\ell(x_{\ell 1}, \dots, x_{\ell k_\ell}))$ in $R(h_1, \dots, h_k)$ by $S(y_1, \dots, y_\ell)$ and then add the quantifier prefix $\exists y_1 \dots \exists y_\ell$ and the conjunct $f_1(x_{11}, \dots, x_{1k_1}) = y_1 \land \dots \land f_\ell(x_{\ell 1}, \dots, x_{\ell k_\ell}) = y_\ell$ to $\psi$ so that it has the following form:
\begin{align*}
\psi = &~ \exists y_1 \dots \exists y_\ell : R(h_1, \dots, h_k) \land \\
& ~f_1(x_{11}, \dots, x_{1k_1}) = y_1 \land \dots \land f_\ell(x_{\ell 1}, \dots, x_{\ell k_\ell}) = y_\ell
\end{align*}
Repeat this process until all arguments for all function and relation symbols are variables.
The resulting formula is semantically equivalent to $\varphi$, since the existentially quantified variables are forced into the same values as in the original formula by the added conjuncts.

For function symbols, this works analogously. \qed
\end{proof}

The proof for the previous Lemma also directly applies to the less general case, where we entirely omit any additional functions and relations:

\begin{lemma}
For every $\FO+\SUM+\PROD$ formula $\varphi$, there is a $\FO+\SUM+\PROD$ formula $\varphi'$ which is equivalent to $\varphi$ but where all function and relation symbols only have variables as their arguments.
\end{lemma}

\begin{definition} 
	Let \Arb denote the set of all finitary relations over \Rinf{} and all functions $f \colon \R^k \to \R$ for $k \in \N$.
\end{definition}

\noindent{}For the upcoming proof we also need some additional notation: For every \FO{} formula $\varphi$ and every variable $x$ let $\varphi[a/x]$ denote $\varphi$ where each occurrence of $x$ is replaced by $a$.
We write $\varphi[a_1/x_1, ..., a_n/x_n]$ to denote several such replacements.

\begin{theorem}
$\FOArb + \SUM + \PROD = \ACO$.
\label{thm-FOArb}
\end{theorem}

\begin{proof}
The proof for this equality follows a similar pattern as the proof for the respective discrete result as presented in \cite{DBLP:books/daglib/0097931}. 

\noindent{}$\FOArb + \SUM + \PROD \subseteq \ACO$: 

\noindent{}The main idea is to show that for any given \FO{} sentence $\varphi$, a circuit family can be constructed which accepts its input if and only if the input encodes an \R-structure that satisfies $\varphi$. 
This is achieved by using addition and multiplication gates to mimic the functionality of existential and universal quantifiers and Boolean connectives and using the available gate types to represent the different kinds of number and index terms that can appear in \FO{} formulae. 
This is a similar basic idea as in the proof in \cite{DBLP:books/daglib/0097931}, however, the technical execution of that idea is quite different thanks to the fact that we are dealing with arithmetic circuits and a logic which deals with Boolean and arithmetic terms of a dyadic structure.

To show that $\FOArb + \SUM + \PROD{}$ is included in \ACO, we will show that for any $\FOArb + \SUM + \PROD$-sentence $\varphi$, we can create an \ACO{} circuit family which decides exactly the set defined by $\varphi$. 
Without loss of generality let $\varphi$ contain only function and relation symbols which only have variables as their arguments.
Given a fixed encoding size $n$ of input \R-structures $\struc{D} = (\struc{A}, \struc{F})$ ($n$ = $\abs{\enc(\struc{D})}$), we can for any \FO{} formula reconstruct $\abs{A}$ from $n$ as described on page~\pageref{par_enc_rec}. 
We will denote $\abs{A}$ by $u$. 

For any subformula $\psi$ of $\varphi$ with exactly $k$ free variables $x_1, ..., x_k$, and any vector $(m_1, ..., m_k) \in A^k$ we can construct an arithmetic circuit $\Cn^{\psi(m_1, ..., m_k)}$ with the following property: 
For any input structure \struc{D} such that $\abs{\enc(\struc{D})} = n$ it holds that $\struc{D} \models \psi[m_1/x_1, ..., m_k/x_k]$ if and only if $\enc(\struc{D})$ is accepted by $\Cn^{\psi(m_1, ..., m_k)}$.

At the very top of the circuit is the output node. 
The rest of the circuit is defined by induction. A formula $\varphi$ with $k$ free variables $x_1, ..., x_k$ and natural numbers $m_1, ..., m_k$, with $1 \leq m_i \leq u$ for all $i$ are given. 
\begin{enumerate}
\item Let $\varphi \fassign \exists y \psi(y)$. If $y$ does not occur free in $\psi$, then the respective circuit for $\varphi$ is the same as for $\psi$, i.e., $\Cn^{\varphi(m_1, ..., m_k)} = \Cn^{\psi(m_1, ..., m_k)}$. Otherwise, the free variables in $\psi$ are $x_1, ..., x_k, y$. $\Cn^{\varphi(m_1, ..., m_k)}$ now consists of a sign gate with an unbounded {fan-in} addition gate as its predecessor which in turn has the circuits $\Cn^{\psi(m_1, ..., m_k, i)}$ as its predecessors for $1 \leq i \leq u$.
\item If $\varphi \fassign \forall y \psi(y)$, then $\Cn^{\varphi(m_1, ..., m_k)}$ is defined as in the existential case, but with a multiplication gate below the sign gate.
\item Let $\varphi \fassign \neg \psi$. Then $\Cn^{\varphi(m_1, ..., m_k)}$ consists of a subtraction gate, which subtracts the sign of $\Cn^{\psi(m_1, ..., m_k)}$ from 1. 
\item Let $\varphi \fassign \psi \land \xi$. Then $\Cn^{\varphi(m_1, ..., m_k)}$ consists of a sign gate followed by a multiplication gate with $\Cn^{\psi(m_1, ..., m_k)}$ and $\Cn^{\xi(m_1, ..., m_k)}$ as its predecessors. (The sign gate is technically not necessary for this case, but we keep it for consistency with e.g. the construction for $\lor$.)
\item If $\varphi \fassign \psi \lor \xi$, $\varphi \fassign \psi \imp \xi$ or $\varphi \fassign \psi \biimp \xi$, then $\Cn^{\varphi(m_1, ..., m_k)}$ follows analogously to $\varphi \fassign \psi \land \xi$. 
\item Let $\varphi \fassign h_1 \feq h_2$ for index terms $h_1, h_2$. Then $\Cn^{\varphi(m_1, ..., m_k)}$ consists of an equality gate with the circuits $\Cn^{h_1(m_1, ..., m_k)}$ and $\Cn^{h_2(m_1, ..., m_k)}$ as its predecessors.
\item If $\varphi \fassign t_1 \feq t_2$ for number terms $t_1, t_2$, then $\Cn^{\varphi(m_1, ..., m_k)}$ is defined analogously to the case with index terms.
\item Let $\varphi \fassign t_1 < t_2$ for number terms $t_1, t_2$. Then $\Cn^{\varphi(m_1, ..., m_k)}$ consists of a $<$ gate with $\Cn^{t_1(m_1, ..., m_k)}$ and $\Cn^{t_2(m_1, ..., m_k)}$ as its predecessors.
\end{enumerate}
For the cases 6, 7 and 8, we also need to show how non-formula index and number terms can be evaluated by our circuit. We will define these by induction as well. Let $h$ be an index term:
\begin{enumerate}
\item Let $h \fassign x$ for $x \in \Vars$. Then $x$ must be $x_i$ for an $i \in {1, ..., k}$ and have previously been quantified. Then $\Cn^{h(m_1, ..., m_k)}$ consists of the constant gate with value $m_i$.
\item Let $h \fassign f(x_1, ..., x_\ell)$ for a $\ell$-ary function symbol $f \in L_a$ and variables $x_1, ..., x_\ell$.
Let $x_i$ be the $\textit{idx}_i$th variable quantified in the original formula for all $1 \leq i \leq \ell$.
Then $\Cn^{h(m_1, ..., m_k)}$ consists of a single constant gate with the value $f(m_{\textit{idx}_1}, \dots, m_{\textit{idx}_\ell})$.
\item If $h \fassign f(x_1, ..., x_\ell)$ for a $\ell$-ary function symbol $f \in L_s$ and variables $x_1, ..., x_\ell$, then $\Cn^{h(m_1, ..., m_k)}$ consists of a single addition gate which has the input gate representing the value $f(m_{\textit{idx}_1}, \dots, m_{\textit{idx}_\ell})$ as above as its sole predecessor.
(The addition gate itself only serves the purpose of making the description of the uniformity of the circuit a little bit easier later on.
It can be regarded as a dummy gate.)
We know where the correct input gate is, since we know the ordering and arities of the function symbols in the input structure.
\end{enumerate}
Let $t$ be a number term:
\begin{enumerate}
\item If $t \fassign c$ for $c \in \R$, then $\Cn^{t(m_1, ..., m_k)}$ consists of a constant gate with value $c$.
\item If $t \fassign f(h_1, ..., h_\ell)$ for a $\ell$-ary function symbol $f \in L_a$ and index terms $h_1, ..., h_\ell$, then $\Cn^{t(m_1, ..., m_k)}$ is defined analogously to the second case of defining index terms.
\item If $t \fassign f(h_1, ..., h_\ell)$ for a $\ell$-ary function symbol $f \in L_f$ and index terms $h_1, ..., h_\ell$, then $\Cn^{t(m_1, ..., m_k)}$ is defined as above but with the input gates describing $f$ instead of constant gates.
\item If $t \fassign t_1 + t_2$ or $t \fassign t_1 \times t_2$ for number terms $t_1, t_2$, then $\Cn^{t(m_1, ..., m_k)}$ consists of a + or $\times$ gate at the top with the circuits $\Cn^{t_1(m_1, ..., m_k)}$ and $\Cn^{t_2(m_1, ..., m_k)}$ as its predecessors.
\item If $t \fassign \sumi(t_1(i))$ or $t \fassign \prodi(t_2(i))$, then the circuit is constructed as for the existential or universal quantifier, respectively, except that the sign gate is omitted.
\end{enumerate}
If $\varphi$ is a sentence, then this construction leads to a circuit deciding $S = \{\struc{D} \in \structR(\sigma) \mid \struc{D} \models \varphi \}$. 
Since this circuit's depth does not depend on $n$ and its size is polynomial in $n$, $S \in \ACO$.

\noindent{}$\ACO \subseteq \FOArb + \SUM + \PROD$:

\noindent{}The idea for this inclusion is to construct a formula for a given circuit family $\C$ that is satisfied by exactly those structures whose encodings are evaluated to $1$ by the circuits of  \C{}. 
This is accomplished by defining number terms which encode the structure of the given circuit. 

This idea is again very similar to the proof in \cite{DBLP:books/daglib/0097931}, nevertheless, again the differences lie in the technical details. 
While the structures used in \cite{DBLP:books/daglib/0097931} are word structures, the functional structures used here require interpreting the circuit inputs as an encoded \R-structure which contains a single unary function that maps an index $i$ to the value of the $i$th input gate of the circuit. 
These real values then need to be accumulated and "carried" through the circuit by defining a number term for each level of the circuit, which maps each gate on that level to its value during the computation.

To show that \ACO{} is included in $\FOArb + \SUM + \PROD$, we create, for any given \ACO{} set $S$, an $\FOArb + \SUM + \PROD$-sentence which defines $S$. In order to achieve this, we want to create a sentence, which talks about the structure of the circuits of the \ACO-circuit family which decides $S$. Since we have access to arbitrary functions, we can essentially just encode the structure of any given circuit into functions and have the interpretation of the function symbols we use be dependent on the length of the input $n$. However, the function symbols themselves, and thus the formula, do not depend on $n$. 
Since the depth of our circuits is constant and we can assume that they are tree-like with each input-output-path having the same length, as shown in Lemma \ref{lem_treelike}, we can construct a sentence which essentially describes the gates on each {level} of the circuit.
Let $S \in \ACO$ via circuit family \C, $depth(\Cn)=d$ and let $q$ be such that $size(\Cn) \leq n^q$ for all $n \in \N$. Without loss of generality, let $\Cn$ be a circuit as described in Lemma~\ref{lem_treelike}, i.e., for every gate $g$ in $\Cn$ it holds that all paths from input gates to $g$ have the same length. We now create a $\FOArb + \SUM + \PROD$-sentence $\varphi$ which defines the set decided by \C. The set $L_f$ of the signature of $\varphi$ will only contain one function symbol $f$, which then for every input gate $v$ in $\Cn$ leads to $f(v)$ being interpreted as the value of $v$. 
Since \Cn{} is of size at most $n^q$, we can uniquely identify the gates of \Cn{} with elements of $A^q$. Let $v$ be a gate in $\Cn$ encoded by $(v_1, ..., v_q)$. $t_n \colon A^q \rightarrow \R$, $c_n \colon A^q \rightarrow \R$, $in_n \colon A^{q+1} \rightarrow \R$ and $pred_n \colon A^{2q} \rightarrow \R$ are functions where $t_n(v_1, ..., v_q)$ is the type of $v$ as per Definition~\ref{def_arith_circ}, $in_n(v_1, ..., v_q, i)$ is $1$ if $v$ is the input gate $i$ of \Cn{} and $0$ otherwise, $c_n(v_1, ..., v_q)$ is the value of gate $v$ if $v$ is a constant gate or $0$, if it is not and $pred_n(v_1, ..., v_q, w_1, ..., w_q)$ is $1$ if $v$ is a predecessor of the gate encoded by $(w_1, ..., w_q)$ and $0$ otherwise. We will use $t$, $in$, $c$ and $pred$ as the respective symbols for these functions. Note that this means that the interpretation of these symbols depends on the input structure.
We can now create a $q$-ary number term $val_x(v_1, ..., v_q)$ for every $x \leq d$, such that it holds that if $(v_1, ..., v_q)$ encodes a gate in $\Cn$ on {level} $x$ (meaning that every path from an input gate to $v$ has length $x$) then for all inputs $(a_1, ..., a_n)$ to the circuit $\Cn$, $val_x(v_1, ..., v_q)$ is the value of the gate encoded by $(v_1, ..., v_q)$ in $\Cn$'s computation when given an \R-structure $\struc{D}$ where $\enc(\struc{D}) = (a_1, ..., a_n)$. 
We will define $val_x$ by induction on $x$. If $x = 0$ then $(v_1, ..., v_q)$ must encode an input gate. We therefore have:
\begin{equation}
val_0(v_1, ..., v_q) \fassign \sumi(in(v_1, ..., v_q, i) \times f(i))
\end{equation}
For $1 \leq x \leq d$, define $val_x$ as follows:
\begin{equation}
\begin{split}
val_x(v_1, ..., v_q) \fassign & \chi[t(v_1, ..., v_q) \feq 2] \times T_{2,x}(v_1, ..., v_q)\\
& + \chi[t(v_1, ..., v_q) \feq 3] \times T_{3,x}(v_1, ..., v_q)\\
& + \chi[t(v_1, ..., v_q) \feq 4] \times T_{4,x}(v_1, ..., v_q)\\
& + \chi[t(v_1, ..., v_q) \feq 5] \times T_{5,x}(v_1, ..., v_q)\\
& + \chi[t(v_1, ..., v_q) \feq 6] \times T_{6,x}(v_1, ..., v_q)
\end{split}
\end{equation}
where 
\begin{align}
T_{2,x}(v_1, ..., v_q) &\fassign c(v_1, ..., v_q)\\
T_{3,x}(v_1, ..., v_q) &\fassign \sumiq(pred(i_1, ..., i_q, v_1, ..., v_q) \times val_{x-1}(i_1, ..., i_q))\\
T_{4,x}(v_1, ..., v_q) &\fassign \prodiq(pred(i_1, ..., i_q, v_1, ..., v_q) \times val_{x-1}(i_1, ..., i_q))\\
T_{5,x}(v_1, ..., v_q) &\fassign \sumiq(pred(i_1, ..., i_q, v_1, ..., v_q) \times sign(val_{x-1}(i_1, ..., i_q)))\\
T_{6,x}(v_1, ..., v_q) &\fassign \sumiq(pred(i_1, ..., i_q, v_1, ..., v_q) \times val_{x-1}(i_1, ..., i_q))
\end{align}
We can now use $val_x$ to define a formula $\varphi$ over signature $\{\{\}, \{f\}, \{t, in, c, pred\}\}$ which defines the set decided by $\Cn$ as follows: (Recall that $d$ denotes the depth of the circuits of the circuit family defining $S$.)
\begin{equation}
\varphi \fassign \forall i_1 ... \forall i_q(\chi[t(i_1, ..., i_q) \feq 6 \imp val_{d}(i_1, ..., i_q) \feq 1])
\end{equation}
The formula $\varphi$ is independent of the input length $n$, however the interpretations of its function symbols of $L_a$ are not. \qed
\end{proof}

\section{A Characterization for \texorpdfstring{\unip{}{\normalfont -}\ACO}{UP-AC0}}
\label{main-up}

Having now developed a description for non-uniform \ACO, in the upcoming part of this paper we derive descriptions for several of its uniform variations.
In particular, we are going to have a look at two uniform subclasses and one generalization of \ACO based on time complexity of \R-machines and one uniform variation based on logical descriptions. 
We start by giving a description for the polynomial time uniform \unip{}-\ACO. 

For this reason, we introduce another notation here:

\begin{definition}
By $\FTIME{}(f(n))$ we will denote all functions that for a finite set $S$ and $k \in \N$ map from $S^k$ to \R{} or to $S$ and that are computable by an \R-machine in time bounded by $\bO(f(\abs{S}))$.
\end{definition}

\begin{theorem}
\label{thm_U_P-AC0}
$\FO[\FTIME{}(n^{\bO(1)})] + \SUM + \PROD = \unip\text{-}\ACO$
\end{theorem}
\begin{proof}
$\FO[\FTIME{}(n^{\bO(1)})] + \SUM + \PROD \subseteq \unip\text{-}\ACO$:

\noindent{}The construction of the circuit is analogous to the one in Theorem~\ref{thm-FOArb}. We now need to demonstrate that the constructed circuit is \ptr{}-uniform. This follows from the fact that the circuit's size is polynomial in the length of its input $n$ and that the construction of each gate takes at most polynomial time. 
In fact, the time it takes to construct the next gate when constructing the circuit in, for example, a depth-first manner is constant in all cases except for those, in which a function or a predicate of $L_a$ needs to be evaluated.
In those cases, the required time is polynomial. That means that the entire circuit can be constructed in polynomial time. We will choose as the numbering of the circuit just the order, in which the gates are first constructed. Since we can compute $\abs{A}$ from $n = \abs{\enc(\struc{D})}$ in logarithmic time as described on page~\pageref{par_enc_rec}, it follows that there exists a machine which on input $(n, v_{nr}, p_{idx})$ can compute $(t, p_{nr}, c)$ as described on page~\pageref{par_uniformity} in time bounded by a polynomial in $n$. \smallskip

\noindent{}$\unip\text{-}\ACO \subseteq \FO{}[\FTIME(n^{\bO(1)})] + \SUM + \PROD$: 

\noindent{}For a given $\unip\text{-}\ACO$ set $S$, we can also create a formula in the same way as in Theorem \ref{thm-FOArb}. We only need to show that we can define the required number terms $t(v_1, ..., v_q)$, $c(v_1, ..., v_q)$, $in(v_1, ..., v_q, i)$, $pred(v_1, ..., v_q, w_1, ..., w_q)$, $\sumi(F(i_1, ..., i_q, \ow))$ and $\prodiq(F(i_1, ..., i_q, \ow))$ in $\FO{}[\FTIME(n^{\bO(1)})] + \SUM + \PROD$, since we can then just use the construction from Theorem~\ref{thm-FOArb}. Let $A$ be the universe of the input structure.
\begin{enumerate}
\item Since the family defining $S$ is \ptr{}-uniform, clearly $t(v_1, ..., v_q)$, $c(v_1, ..., v_q)$ and $in(v_1, ..., v_q, i)$ can be defined in $\FO{}[\FTIME(n^{\bO(1)})] + \SUM + \PROD$.
\item $\sumi(F(i_1, ..., i_q, \ow))$ and $\prodi(F(i_1, ..., i_q, \ow))$ are given by the extension.
\item $pred(v_1, ..., v_q, w_1, ..., w_q)$ can be defined in $\FO{}[\FTIME(n^{\bO(1)})] + \SUM + \PROD$ in the following way:
We define the predicate
	\begin{equation}
	pred_k \defeq \Set{(v_1, ..., v_q, w_1, ..., w_q, k_1, ..., k_q) | \parbox{4.5cm}{$v$ is the $k$th predecessor of $w$ where $v$ is the gate encoded by $(v_1, ..., v_q)$, $w$ is the gate encoded by $(w_1, ..., w_q)$ and $k$ is the number encoded by $(k_1, ..., k_q)$.}} 
	\end{equation}
which we can evaluate in polynomial time, since $S$ is \ptr{}-uniform. We can now define $pred(v_1, ..., v_q, w_1, ..., w_q)$ in $\FO{}[\FTIME(n^{\bO(1)})] + \SUM + \PROD$ as follows:
	\begin{multline}
	pred(v_1, ..., v_q, w_1, ..., w_q) \fassign \\ \chi[\exists k_1, ..., \exists k_q : pred_k(v_1, ..., v_q, w_1, ..., w_q, k_1, ..., k_q)]
	\end{multline}
\end{enumerate}
Therefore we can define $S$ using a $\FO{}[\FTIME(n^{\bO(1)})] + \SUM + \PROD$-sentence. \qed
\end{proof}


\section{A Characterization for \texorpdfstring{\unil{}{\normalfont -}\ACO}{UL-AC0}}
\label{main-ult}

We have demonstrated that the same construction as in the proof of Theorem \ref{thm-FOArb} can be applied in the \ptr{}-uniform case if we restrict our logic to a polynomial extension rather than a universal one. 
For the second uniformity result, we will produce a description for \unil{}-\ACO{} sets. 
The construction is again very similar to the one for the non-uniform case.

We take advantage of the tree-like nature of the circuits we constructed with our method so far and number their gates in a post-order fashion. 
This will be helpful for showing \ltr{}-uniformity, since it gives us the path from the output gate to any other gate and hence allows us to construct it without needing to construct the entire circuit.
Essentially, a tree-like circuit would be a tree if we copied its input gates for each outgoing edge, such that each input gate has outdegree $1$.

\begin{definition}
By the \new{\fsize{}} of a circuit we denote the number of gates of the circuit where each input gate is counted once for each connection it has to the circuit. 
An example for this is given in Figure~\ref{fig_tree-shape-size}.
\end{definition}

\begin{figure}
\begin{center}
\begin{tikzpicture}[
	base/.style={circle,draw,minimum size=30pt}
]
\node[base] (out) {$out$}
	child {node[base] (+1) {$+$}
		child {node[base] (t1) {$\times$}
			child {node[base] (in1) {$in_1$} }
			child {node[base] (in2) {$in_2$} }
		}
		child {node[base] (+2) {$+$} 
			child[missing]
			child {node[base] (in3) {$in_3$}}
		}
	};
\draw [-] (in2) -- (+2); 
\end{tikzpicture}
\end{center}
\caption{The size of this circuit is 7, however its \fsize{} is 8, since the second input gate has two successors.}
\label{fig_tree-shape-size}
\end{figure}

\noindent{}Since we would like to have access to the \fsize{} of our circuits during our computations, we need to see, how efficiently we can compute the \fsize{} of circuits in our construction. As it turns out, the number of computation steps we need does not depend on the size of our given input structure and is therefore constant for our purposes.

\begin{lemma}
\label{lem_subcircuit_fullsize}
For a circuit constructed for a given $\FOR$-sentence and \R-structure, as in Theorem~\ref{thm-FOArb}, we can compute the \fsize{} of the circuit for $\varphi$ or any circuit for a subformula or number or index term of $\varphi$ in constant time with respect to the given input structure.
\end{lemma}
\begin{proof}
Note that since the variable assignments of the notation for Theorem~\ref{thm-FOArb} do not make a difference for the size of the circuit, we will omit them in this proof.

Let $\varphi$ be the given formula and $u = \abs{A}$ be the size of the input structure. We give the \fsize{} of the circuit for every subformula and term of $\varphi$ by induction in the same way, as the circuit is constructed in the proof of Theorem~\ref{thm-FOArb}.

\begin{enumerate}
\item Let $\varphi = \exists y \psi(y)$. Then $\fsize(\C_n^{\varphi}) = u \cdot \fsize(\C_n^{\psi})$.
\item If $\varphi = \forall y \psi(y)$, then the \fsize{} is computed as in the existential case.
\item Let $\varphi = \neg \psi$. Then $\fsize(\C_n^{\varphi}) = 2 + \fsize(\C_n^{\psi})$.
\item Let $\varphi = \psi \land \xi$. 
Then 

$\fsize(\C_n^{\varphi}) = 1 + \fsize(\C_n^{\psi}) + \fsize(\C_n^{\xi})$.
\item If $\varphi = \psi \lor \xi$, $\varphi = \psi \imp \xi$ or $\varphi = \psi \biimp \xi$, then the \fsize{} can be computed analogously to $\varphi = \psi \land \xi$.
\item Let $\varphi = h_1 \feq h_2$ for index terms $h_1, h_2$. Then $\fsize(\C_n^{\varphi}) = 1 + \fsize(\C_n^{h_1}) + \fsize(\C_n^{h_2})$.
\item If $\varphi = t_1 \feq t_2$ for number terms $t_1, t_2$, then the \fsize{} can be computed as for index terms.
\item If $\varphi = t_1 < t_2$ for number terms $t_1, t_2$, then the \fsize{} can be computed as for equality.
\end{enumerate}

\noindent{}The \fsize{} of index terms is computed as follows. Let $h$ be an index term.

\begin{enumerate}
\item Let $h = x$ for $x \in \Vars$. Then $\fsize(\C_n^{h}) = 1$.
\item If $h \fassign f(x_1, ..., x_\ell)$ for a $\ell$-ary function symbol $f \in L_s$ and variables $x_1, ..., x_\ell$, then $\fsize(\C_n^{h}) = 1$.
\item If $h = f(x_1, ..., x_\ell)$ for a $\ell$-ary function symbol $f \in L_s$ and variables $x_1, ..., x_\ell$,, then $\fsize(\C_n^{h}) = 2$.
\end{enumerate}

\noindent{}The \fsize{} of number terms is computed as follows. Let $t$ be a number term.

\begin{enumerate}
\item Let $t = c$ for $c \in \R$. Then $\fsize(\C_n^{t}) = 1$.
\item If $t = f(x_1, ..., x_\ell)$ for a $\ell$-ary function symbol $f \in L_a$ and variables $x_1, ..., x_\ell$, then $\fsize(\C_n^{t}) = 1$.
\item If $t = f(x_1, ..., x_\ell)$ for a $\ell$-ary function symbol $f \in L_f$ and variables $x_1, ..., x_\ell$, then $\fsize(\C_n^{t}) = 2$.
\item Let $t = t_1 + t_2$ or $t = t_1 \times t_2$ for number terms $t_1, t_2$, then $\fsize(\C_n^{t}) = 1 + \fsize(\C_n^{t_1}) + \fsize(\C_n^{t_2})$.
\end{enumerate}

To get the \fsize{} of the entire circuit for $\varphi$, we need to add $1$ to the final \fsize{}, since the output gate is not considered for subformulas.

We have now shown how to compute the \fsize{} of the circuit for $\varphi$ and any of its subformulas and terms. Each individual computation can be done in constant time, and since the formula is constant, only a constant amount of those operations is required. \qed
\end{proof}

\begin{theorem}
\label{thm_U_L-AC0}
$\FO[\FTIME{}(\log n)] + \SUM + \PROD = \unil\text{-}\ACO$
\end{theorem}
\begin{proof}
$\FO{}[\FTIME(\log n)] + \SUM + \PROD \subseteq \unil\text{-}\ACO$:

\noindent{}Just as in the polynomial case, we will use the same construction as in Theorem \ref{thm-FOArb} for the logarithmic case. 
Showing the \ltr{}-uniformity of the resulting circuit, however, is not as simple as it was in Theorem \ref{thm_U_P-AC0}, since we cannot just construct the entire circuit to retrieve the information for a singular gate. 
We can, however, construct only part of the circuit to arrive at the gate which we would like to retrieve in order to remain within logarithmic time. 
We will essentially construct the path from the output node to the node we are looking for, which has constant length.
Let $S$ be the set of \R-structues defined by a given $\FO{}[\FTIME(\log n)] + \SUM + \PROD$-sentence $\varphi$. 
To create a circuit family deciding $S$, define the structure of our circuit depending on $\varphi$ similarly to the proof of Theorem \ref{thm-FOArb}. 
Here however, we will make sure that for each gate $v$, we know the \fsize{} of all of its direct subcircuits, i.e., the subcircuits induced by $v$'s predecessor gates, in order to make sure that we continue our construction at the right predecessor of $v$. 
In doing so, we can always compute the \fsize{}s of the predecessor subcircuits of any given node in time constant in the length of the input. 
We will additionally number our nodes in post-order, to ensure that we know where to continue constructing our circuit. 
The circuit is then constructed/structured as follows:

Since the input gates do not behave tree-like, we explicitly give the numbering they get, whenever it is needed: 
The $i$th input gate has the number $\fsize(\Cn)+i$. 
At the very top of the circuit, there is the output node numbered $\fsize(C_n)$, the predecessor of which then has the number $\fsize(C_n)-1$. 
The rest of the circuit is numbered as follows: 
Let the root gate of the subcircuit representing $\varphi$ be numbered $q$.
\begin{enumerate}
\item Let $\varphi \fassign \exists y \psi(y)$. 
Then the construction is as in Theorem \ref{thm-FOArb}. 
The sign gate is numbered $q$, the addition gate is numbered $q-1$ and the root of the $i$th predecessor circuit $\Cn^{\psi(m_1, ..., m_k, i)}$ is numbered $q-2-(u-i)\cdot\fsize(\Cn^{\psi(m_1, ..., m_k, 1)})$. 
(Since the \fsize of $\Cn^{\psi(m_1, ..., m_k, i)}$ is the same for all $i$, we can simply use the \fsize of $\Cn^{\psi(m_1, ..., m_k, 1)}$ for each $i$.)
\item If $\varphi \fassign \forall y \psi(y)$, then the construction is as in Theorem \ref{thm-FOArb} and the numbering is analogous to the existential case.
\item Let $\varphi \fassign \neg \psi$. Then the construction is as in Theorem \ref{thm-FOArb}, the subtraction gate is numbered $q$, the constant gate with value $1$ is numbered $q-1$ and the root of $\Cn^{\psi(m_1, ..., m_k)}$ is numbered $q-2$.
\item Let $\varphi \fassign \psi \land \xi$. Then the construction is as in Theorem \ref{thm-FOArb}, the sign node is numbered $q$, the $\times$ gate is numbered $q-1$, the root of $\Cn^{\psi(m_1, ..., m_k)}$ is numbered $q-2-\fsize(\Cn^{\xi(m_1, ..., m_k)})$ and the root of $\Cn^{\xi(m_1, ..., m_k)}$ is numbered $q-2$. 
\item If $\varphi \fassign \psi \lor \xi$, $\varphi \fassign \psi \imp \xi$ or $\varphi \fassign \psi \biimp \xi$, then $\Cn^{\varphi(m_1, ..., m_k)}$ and its numbering follows analogously to $\varphi \fassign \psi \land \xi$.
\item Let $\varphi \fassign h_1 \feq h_2$ for index terms $h_1, h_2$. Then the construction is as in Theorem \ref{thm-FOArb}, the equality gate is numbered $q$, the root of $\Cn^{h_1(m_1, ..., m_k)}$ is numbered $q-1-\fsize(\Cn^{h_2(m_1, ..., m_k)})$ and the root of $\Cn^{h_1(m_1, ..., m_k)}$ is numbered $q-1$.
\item If $\varphi \fassign t_1 \feq t_2$ for number terms $t_1, t_2$, then $\Cn^{\varphi(m_1, ..., m_k)}$ is defined and numbered analogously to the case with index terms.
\item Let $\varphi \fassign t_1 < t_2$ for number terms $t_1, t_2$. Then $\Cn^{\varphi(m_1, ..., m_k)}$ is defined as in Theorem \ref{thm-FOArb} and numbered analogously to the case of equality.
\end{enumerate}
For the cases 6, 7 and 8, we also need to show how construction and numbering can be done for non-formula index and number terms. We will define these by induction as well. Let $h$ be an index term:
\begin{enumerate}
\item Let $h \fassign x$ for $x \in \Vars$. Then the construction is as in Theorem \ref{thm-FOArb} (The constant gate is numbered $q$).
\item Let $h \fassign f(x_1, ..., x_\ell)$ for a $\ell$-ary function symbol $f \in L_a$ and variables $x_1, ..., x_\ell$.
Then the construction is as in Theorem~\ref{thm-FOArb} and the singular constant gate is numbered $q$.
\item If $h \fassign f(x_1, ..., x_\ell)$ for a $\ell$-ary function symbol $f \in L_s$ and variables $x_1, ..., x_\ell$, then the construction is as in Theorem \ref{thm-FOArb}, the addition "dummy" gate is numbered $q$ and the input gate is numbered according to the rule at the top. 
\end{enumerate}
Let $t$ be a number term:
\begin{enumerate}
\item If $t \fassign c$ for $c \in \R$, then the construction is as in Theorem \ref{thm-FOArb}.
(The constant gate is numbered $q$.)
\item If $t \fassign f(x_1, ..., x_\ell)$ for a $\ell$-ary function symbol $f \in L_a$ and variables $x_1, ..., x_\ell$, then the construction is as in Theorem \ref{thm-FOArb} and the numbering is done as described in the case of index terms.
\item If $t \fassign f(x_1, ..., x_\ell)$ for a $\ell$-ary function symbol $f \in L_f$ and variables $x_1, ..., x_\ell$, then the construction is as in Theorem \ref{thm-FOArb} and the numbering is done as described in the case of index terms.
\item If $t \fassign t_1 + t_2$ or $t \fassign t_1 \times t_2$ for number terms $t_1, t_2$, then the construction is as in Theorem \ref{thm-FOArb}. The addition gate is numbered $q$, the root of $\Cn^{t_2{m_1, ..., m_k}}$ is numbered $q-1$ and the root of $\Cn^{t_1(m_1, ..., m_k)}$ is numbered $q-1-\fsize(\Cn^{t_2(m_1, ..., m_k)})$.
\item If $t \fassign \sumi(t_1(i))$ for a number term $t_1$ in which $i$ occurs freely, then $\Cn^{t(m_1, ..., m_k)}$ consists of an addition gate at the top, numbered $q$, with the root nodes of the circuits $\Cn^{t_1(m_1, ..., m_k, i)}, 1 \leq i \leq u$ as its $u$ predecessors, numbered $q - 1 - (u-i) \cdot \fsize(\Cn^{t_1(m_1, ..., m_k, i)})$, similar to the case of existential quantification. If $i$ does not occur freely in $t_1$, then the predecessors of node $q$ are $u$ gates,  of which each induces a copy of the circuit $\Cn^{t_1(m_1, ..., m_k)}$ and which are numbered the same way as for the case where $i$ is free in $t_1$.
\item If $t \fassign \prodi(t_1(i))$ for a number term $t_1$, then the construction and numbering is done as above, just using a multiplication gate instead of an addition gate.
\end{enumerate}
Note that this numbering gives each gate a distinct number and makes sure that for all non-input gates $v$ it holds that $v$'s number is higher than those of $v$'s predecessors. Additionally, it holds that for any two predecessors $v_1$ and $v_2$ of $v$, if $v_1$ is numbered lower than $v_2$, then all nodes in $v_1$'s induced subcircuit are also numbered lower than $v_2$ and vice versa. Since we can compute the \fsize{} of any subcircuit in constant time, we can also compute the number of the node where we need to continue in constant time. Note also that since the input gates do not behave tree-like there are holes in the numbering. 
Now we define an \R-machine $M$ which on input $(n, v_{nr}, p_{idx})$ returns $(t, p_{nr}, c)$ as described on page~\pageref{par_uniformity}. As described on page~\pageref{par_enc_rec}, we know that we can compute $\abs{A}$ from $n = \abs{\enc(\struc{D})}$ in time logarithmic in $n$. To now produce the desired output, we take advantage of our node numbering. We know that our last node -- the output node -- has number $\fsize(\Cn)$ and its singular predecessor node has number $\fsize(\Cn) - 1$. Let $\fs_\varphi$ denote $\fsize(\Cn^{\varphi(m_1, ..., m_k)})$ -- which is the same as $\fsize(\Cn^{\varphi(m_1, ..., m_k, i)})$ etc., since the variable assignments do  not have an effect on the size of the circuit -- and $u$, as in the proof of Theorem \ref{thm-FOArb}, the size of the universe of the input structure $\abs{A}$. Now the machine works as follows: 
If $v_{nr} > \fsize{\Cn} + n$, then return $(0,0,0)$.
If $v_{nr} = \fsize(\Cn)$ then return $(6, \fsize(\Cn)-1, 0)$ if $p_{idx} = 1$ and $(6, 0, 0)$ otherwise.
If $v_{nr} = \fsize(\Cn) + i$, for $i \in \lbrace 1, ..., n \rbrace$ then return $(1, 0, i)$.
Otherwise proceed as follows: Let $q$ be the number of the root of the current subcircuit. (We will use $q$ to describe both the value $q$ and the register in which we store that value.)\medskip
\begin{enumerate}
\item Let $\varphi \fassign \exists y \psi(y)$.
If $v_{nr}$ = $q$, then return $(5, q-1, 0)$ if $p_{idx} = 1$ and $(5, 0, 0)$ otherwise.
If $v_{nr}$ = $q-1$, then return $(3, q - 2 - (u-p_{idx})\cdot\fs_\psi, 0)$ if $p_{idx} \leq u$ and $(5, 0, 0)$ otherwise.
Otherwise gate $v_{nr}$ is contained in the subcircuit induced by the gate numbered $y = q-2-(\ceil\cdot{\frac{q-1-v_{nr}}{\fs_\psi}}-1) \cdot \fs_\psi$ where $y$ is the smallest natural number such that $y \geq v_{nr}$ and $y = q-2-(u-i)\cdot\fs_\psi$ for some $i \in \{1, ..., u\}$. We can compute $y$ in time logarithmic in $u$ by using binary search on $i$. We therefore store $y$ in $q$ and continue with the construction of the subcircuit induced by node $y$.
\item If $\varphi \fassign \forall y \psi(y)$, then the construction is analogous to the existential case.
\item Let $\varphi \fassign \neg \psi$.
If $v_{nr} = q$, then return $(7, q - 1, 0)$ if $p_{idx} = 1$, $(7, q - 2, 0)$ if $p_{idx}=2$ and $(7, 0, 0)$ otherwise. 
If $v_{nr} = q - 1$, then return $(2, 0, 1)$. 
Otherwise, store $q - 2$ in $q$ and continue with the construction of $\Cn^{\psi(m_1, ..., m_k)}$.
\item Let $\varphi \fassign \psi \land \xi$. 
If $v_{nr} = q$ then return $(5, q-1, 0)$ if $p_{idx}=1$ and $(5, 0, 0)$ otherwise.
If $v_{nr} = q-1$ then return $(4, q - 2 - \fsize(\Cn^{\xi(m_1, ..., m_k)}), 0)$ if $p_{idx} = 1$, $(4, q - 2, 0)$ if $p_{idx} = 2$ and $(4, 0, 0)$ otherwise.
Otherwise, if $v_{nr} \leq q - 2 - \fsize(\Cn^{\xi(m_1, ..., m_k)})$, store $q - 2 - \fsize(\Cn^{\xi(m_1, ..., m_k)})$ in $q$ and construct $\Cn^{\psi(m_1, ..., m_k)}$ and otherwise store $q - 2$ in $q$ and construct $\Cn^{\psi(m_1, ..., m_k)}$.
\item If $\varphi \fassign \psi \lor \xi$, $\varphi \fassign \psi \imp \xi$ or $\varphi \fassign \psi \biimp \xi$, then proceed analogously to $\varphi \fassign \psi \land \xi$.
\item If $\varphi \fassign h_1 \feq h_2$ for index terms $h_1, h_2$, then proceed analogously to the Boolean connectives.
\item If $\varphi \fassign t_1 \feq t_2$ for number terms $t_1, t_2$, then proceed analogously to the Boolean connectives.
\item If $\varphi \fassign t_1 < t_2$ for number terms $t_1, t_2$, then proceed analogously to the Boolean connectives.
\end{enumerate}
For the cases 6, 7 and 8, we also need to explain how to construct the subcircuits for non-formula index and number terms. We will define these by induction as well. Let $h$ be an index term:
\begin{enumerate}
\item Let $h \fassign x$ for $x \in \Vars$. Then if $q = v_{nr}$, $x$ must be $x_i$ for some $x_i \in \Vars$, thus return $(2, 0, m_{i_{xi}})$, where $x_i$ is the $i_{xi}$th quantified variable in the original formula. 
Otherwise return $(0, 0, 0)$.
\item Let $h \fassign f(x_1, ..., x_\ell)$ for a $\ell$-ary function symbol $f \in L_a$ and variables $x_1, ..., x_\ell$, then if $q = v_{nr}$ return $(2, 0, f(m_{i_{x1}}, ..., m_{i_{x\ell}}))$ where $x_i$ was the $i_xi$th quantified variable in the original formula and return $(0, 0, 0)$, otherwise.
\item Let $h \fassign f(x_1, ..., x_\ell)$ for a $\ell$-ary function symbol $f \in L_s$ and variables $x_1, ..., x_\ell$.
If $v_{nr} = q$, then return $(3, \mathrm{input\_nr}, 0)$ where $\mathrm{input\_nr}$ is the number of the input gate representing the function value $f(m_{i_{x1}}, ..., m_{i_{x\ell}})$ with $m_{i_{xj}}$ as above if $p_{idx} = 1$ and $(3, 0, 0)$, if $p_{idx} \neq 1$.
If $v_{nr}$ was the number of an input gate, it would have been returned at the top.
\end{enumerate}
Let $t$ be a number term:
\begin{enumerate}
\item If $t \fassign c$ for $c \in \R$. Then if the constant gate is numbered $v_{nr}$, return $(2, 0, c)$. Otherwise return $(0,0,0)$.
\item If $t \fassign f(x_1, ..., x_\ell)$ for a $\ell$-ary function symbol $f \in L_a$ and variables $x_1, ..., x_\ell$, then construct analogously to the case of index terms.
\item If $t \fassign f(x_1, ..., x_\ell)$ for  $\ell$-ary function symbol $f \in L_f$ and variables $x_1, ..., x_\ell$, then construct analogously to the case of index terms.
\item If $t \fassign t_1 + t_2$ or $t \fassign t_1 \times t_2$ for number terms $t_1, t_2$, continue constructing as in the case of Boolean connectives.
\item Let $t \fassign \sumi(t_1(i))$ for a number term $t_1$.
If $v_{nr} = q$ then return $(3, q-1-(u-p_{idx})\cdot\fsize(\Cn{t_1(m_1, ..., m_k, p_{idx})}), 0)$ if $p_{idx} \leq u$ and $(3, 0, 0)$ otherwise.
Otherwise store $q-1-(u-p_{idx})\cdot\fsize(\Cn{t_1(m_1, ..., m_k, p_{idx})})$ in $q$ and continue with the construction of $\Cn^{t_1(m_1, ..., m_k, p_{idx})}$.
\item If $t \fassign \prodi(t_1(i))$ for a number term $t_1$, then continue constructing as in the case of $\sumi$.
\end{enumerate}
The way $M$ works, after decoding the input structure, it only ever needs to perform a constant number of operations on each level of the circuit, with the exception of the predicates and functions which are not given in the input structure. For those, $M$ needs logarithmic time. This means in total that since the circuit only has constant depth and hence a constant number of levels, $M$ works in logarithmic time. Therefore, $S$ is an element of $\unil$-\ACO. \smallskip

\noindent{}$\unil\text{-}\ACO \subseteq \FO{}[\FTIME(\log n)] + \SUM + \PROD$: 

\noindent{}Showing that a set $S \in \unil\text{-}\ACO$ can be defined using $\FO{}[\FTIME(\log n)] + \SUM + \PROD$ is done in the same way as it was done in the polynomial case (Theorem \ref{thm_U_P-AC0}). We construct the formula analogously and we can compute the functions we need for that construction in logarithmic time as follows:
\begin{enumerate}
\item We can compute the functions $t(v_1, ..., v_q)$, $c(v_1, ..., v_q)$, $in(v_1, ..., v_q, i)$ and $pred(v_1, ..., v_q, w_1, ..., w_q)$ in logarithmic time analogous to Theorem~\ref{thm_U_P-AC0}, since our circuit family is \ltr{}-uniform.
\item $\sumi$ and $\prodi$ are given in the specification of $\FO{}[\FTIME(\log n)] + \SUM + \PROD$.
\end{enumerate}\qed
\end{proof}

\begin{figure}
\begin{center}
\begin{tikzpicture}[
	scale=0.6,
	every node/.style={transform shape}, 
	base/.style={circle,draw,minimum size=30pt}, 
	circ/.style={minimum size=35pt},
	triangle/.style={regular polygon, regular polygon sides = 3, draw, inner sep=0, text width=15mm}
]
\tikzstyle{level 1}=[sibling distance=60mm]
\tikzstyle{level 2}=[sibling distance=40mm]
\tikzstyle{level 3}=[sibling distance=25mm]
\node[base, label={180:$12$}] (out3) {$out$}
	child { node[base, label={180:$11$}] (+3) {$+$} 
		child {node[base, label={180:$10$}] (lt3) {$\times$}
			child {node[base, label={180:$9$}] (c63) {$6$}}
			child {node[base, label={180:$8$}] (i+1) {$+$} 
				child {node[missing] (i13) {} edge from parent[draw=none]}				
			}
		}
		child {node[base, label={180:$7$}] (rt3) {$\times$}
			child {node[base, label={0:$6$}] (elt3) {$\times$} 
				child[missing]
				child {node[base, label={180:$13$}] (in1) {$in_1$}}
			}
			child{node[base, label={180:$3$}] (i+2) {$+$}
				child{node[base, label={0:$14$}] (i23) {$in_2$} }
			}		
		}
	}
	;
\node[base, label={180:$5$}] (ec63) at ($(i+1) !.5! (elt3) + (i13) - (i+1)$) {$6$};
\draw [-] (ec63) -- (elt3);
\draw [-] (in1) -- (i+1);

\end{tikzpicture}
\end{center}
\caption{The circuit from Figure~\ref{fig_circ_ex} has been transformed as shown in Figure~\ref{fig_treelike_ex} and been numbered as in (a simplified version of) the numbering for Theorem~\ref{thm_U_L-AC0}. 
If we for example wanted to construct the addition gate numbered $8$, we would start at $12$, construct the gate $11$ and we would then know to keep going at gate $10$, since $8$ is less than $10$ but larger than $7$. 
Note, that the input gates are exceptions in this numbering, since they do not behave tree-like. 
Their numbering starts just above the \fsize of the circuit, so if a machine producing the gates of the circuit gets a number $12 < k \leq 14$ as an input, it can immediately return $(1, 0, k)$ (as per Definition~\ref{def_uniformity}).}
\label{fig_postorder_ex}
\end{figure}


%

With the construction shown in Theorem~\ref{thm_U_L-AC0} we can now generalize that, whenever we have a variant of \ACO{} given by a time complexity uniformity criterion that is at least logarithmic, we can describe it using first-order logic extended with functions of that class' time complexity and the sum and product rule. This result is formalized as follows:

\begin{corollary}
\label{cor_generalization}
For any function $f \colon \N \rightarrow \N$ with $f(n) \geq \log n$ for all $n$, it holds that
\begin{equation}
\uni{f}\text{-}\ACO = \mathrm{FO}_\R[\FTIME{}(f(n))] + \SUM + \PROD,
\end{equation}
where \uni{f}-$\ACO$ is the class of sets decidable by circuit families, which can be constructed as described in Definition~\ref{def_uniformity} in time bounded by $\bO(f(n))$.

\end{corollary}

\begin{remark}
The logarithmic bound for $f$ in Corollary~\ref{cor_generalization} stems from the time it takes to decode an encoded \R-structure as stated on page~\pageref{par_enc_rec}.

\end{remark}

\section{A Characterization for \texorpdfstring{\unifo{}{\normalfont -}\ACO}{U_FO-AC0_R}}
\label{main-ufo}

The uniformity of the complexity classes we just discussed was based on how much time an \R-machine needs to answer queries about a circuit.
In this section, however, we are going to turn to the connection of first-order logic over the reals and those circuit families, which are themselves describable using first-order formulae and terms.

In order to make the proof for the upcoming theorem a little more concise, we define syntax trees for $\FO+\SUM+\PROD$ formulae.

\begin{definition}\label{def_syntree}
Let $\varphi$ be a $\FO+\SUM+\PROD$ formula or term in which all functions and relations only have variables as their arguments.
The \emph{syntax tree} of $\varphi$ is a tree, which represents the syntactical structure of the the formula or term. 
Each node of that tree represents a syntactical part of $\varphi$, such that each subformula of $\varphi$ is represented by a subtree of $\tree(\varphi)$.
We denote the syntax tree of $\varphi$ by $\tree(\varphi)$.


\begin{enumerate}
\item Let $\varphi = \exists y \psi(y)$ or . Then $\tree(\varphi)$ consists of a node labelled $\exists y$ with an edge to the root of $\tree(\psi(y))$.
\item If $\varphi = \forall y \psi(y)$, then $\tree(\varphi)$ is analogous with an universal quantifier.
\item Let $\varphi = \neg \psi$. Then $\tree(\varphi)$ consists of a node labelled $\neg$ with an edge to the root of $\tree(\psi)$.
\item Let $\varphi = \psi \land \xi$. Then $\tree(\varphi)$ consists of a node labelled $\land$ with an edge to the root of $\tree(\psi)$ and an edge to the root of $\tree(\xi)$.
\item If $\varphi = \psi \lor \xi$, $\varphi = \psi \imp \xi$ or $\varphi = \psi \biimp \xi$, then $\tree(\varphi)$ looks analogous to the case above.
\item Let $\varphi = h_1 \feq h_2$ for index terms $h_1, h_2$. Then $\tree(\varphi)$ consists of a node labelled $=$ with an edge to the root of $\tree(h_1)$ and an edge to the root of $\tree(h_2)$.
\item If $\varphi = t_1 \feq t_2$ for number terms $t_1, t_2$, then $\tree(\varphi)$ looks analogous to the case for index terms.
\item If $\varphi = t_1 < t_2$ for number terms $t_1, t_2$, then $\tree(\varphi)$ looks analogous to the case for equality.
\end{enumerate}

\noindent{}Let $h$ be an index term. Then $\tree(h)$ is defined as follows:

\begin{enumerate}
\item Let $h = x$ for $x \in \Vars$. Then $\tree(h)$ consists only of a single node labelled $x$.
\item Let $h = f(x_1, ..., x_\ell)$ for a $\ell$-ary function symbol $f \in L_s$ and variables $x_1, ..., x_\ell$.
Then $\tree(h)$ consists only of a node labelled $f(x_1, \dots, x_\ell)$.
\end{enumerate}

\noindent{}Let $t$ be an index term. Then $\tree(t)$ is defined as follows:

\begin{enumerate}
\item Let $t = c$ for $c \in \R$. Then $\tree(t)$ consists only of a single node labelled $c$.
\item Let $t = f(x_1, ..., x_\ell)$ for a $\ell$-ary function symbol $f \in L_f$ and variables $x_1, ..., x_\ell$.
Then $\tree(h)$ consists only of a node labelled $f(x_1, \dots, x_\ell)$.
\item Let $t = t_1 + t_2$ or $t = t_1 \times t_2$ for number terms $t_1, t_2$. 
Then $\tree(t)$ consists of a node labelled with the respective arithmetic operation and an edge to $\tree(t_1)$ and an edge to $\tree(t_2)$.
\item Let $t = \sumi(t_1(i))$ for a number term $t_1$. 
Then $\tree(t)$ consists of a node labelled $\sumi$ with an edge to the root of $\tree(t_1)$.
\item Let $t = \prodi(t_1(i))$ for a number term $t_1$. 
Then $\tree(t)$ consists of a node labelled $\prodi$ with an edge to the root of $\tree(t_1)$.
\end{enumerate}

Similarly to our circuits, we will refer to the number of nodes in a syntax tree as the \emph{size} of the syntax tree.
\end{definition}

\begin{example}
The syntax tree for the $\FO$ formula $\exists x : f(x) = 3 \lor g(x) = 2$ is depicted in Figure~\ref{fig_syntree_ex}.

\begin{figure}
\begin{center}
\begin{tikzpicture}[
	scale=0.6,
	every node/.style={transform shape}, 
	base/.style={circle,draw,minimum size=30pt}, 
	circ/.style={minimum size=35pt},
	triangle/.style={regular polygon, regular polygon sides = 3, draw, inner sep=0, text width=15mm}
]
\tikzstyle{level 1}=[sibling distance=60mm]
\tikzstyle{level 2}=[sibling distance=35mm]
\tikzstyle{level 3}=[sibling distance=20mm]
\node[base] (ex) {$\exists x$}
	child { node[base] (or) {$\lor$} 
		child {node[base] (eq1) {$=$}
			child{node[base] (f) {$f(x)$}}
			child{node[base] (3) {$3$}}		
		}
		child {node[base] (eq2) {$=$}
			child{node[base] (g) {$g(x)$}}
			child{node[base] (2) {$2$}}
		}
	}
	;
\end{tikzpicture}
\end{center}
\caption{Syntax tree for the $\FO$ formula $\exists x : f(x) = 3 \lor g(x) = 2$.}
\label{fig_syntree_ex}
\end{figure}
\end{example}

\begin{theorem}\label{thm_UFO_AC0}
$\FO+\SUM+\PROD = \unifo\text{-}\ACO$
\end{theorem}
\begin{proof}
$\FO+\SUM+\PROD \subseteq \unifo\text{-}\ACO$:

In order to prove this inclusion, we need to show that for each set defined by an $\FO+\SUM+\PROD$-sentence, there is also a \FO-uniform circuit family \C deciding it.
Let $\varphi$ be such a sentence.
The circuits of \C are structured in the same way as in Theorem~\ref{thm-FOArb} (except of course for the constructions modelling symbols in $L_a$, which are not necessary here).
It remains to be shown that this circuit family is \FO-uniform.

Let the depth of the circuits of \C be $d$. 
Then for each $n$, each gate $g$ in $C_n$ can uniquely be identified by a sequence of $d$ values in the range $\{1, \dots, n+1\}$, where $n$ is the number of input gates of $C_n$ (and therefore also equal to the encoding length of structures given as inputs to $C_n$).
We will also refer to this sequence as the \emph{number} of $g$.
As a matter of fact, it will turn out that the range $\{1, \dots u+1\}$ is already enough, where $u$ is the size of the universe of the encoded input structure.
This sequence encodes the path from the output gate to $g$ in the following way:

The value at index $i$ in the sequence selects at which child to continue at distance $i-1$ to the output gate. 
The first occurrence of the value $u+1$ in the sequence denotes where to stop.
This means that each number consists of a prefix of values between $1$ and $u$, since the circuit construction of Theorem~\ref{thm-FOArb} only needs at most fan-in $u$, and a suffix consisting of only a sequence of $u+1$s.
Therefore, the output gate is always numbered $(u+1, \dots, u+1)$.
This also means, that the edge relation can easily be deduced from the gate numbers: 
If the prefixes of two gate numbers are identical, except that one of the numbers contians one more non $u+1$ value, there is an edge between the two gates.

In Figure~\ref{fig_unifo_numbering}, the addition gate would be numbered $(1, u+1, u+1, u+1)$ and the multiplication gate would be numbered $(1, 2, u+1, u+1)$.
Since their two numbers differ only in the last non $u+1$ value of the multiplication gate, at the index of which in the number of the addition gate there is a $u+1$, an edge connects these two gates.

\begin{figure}
\begin{center}
\begin{tikzpicture}[
	scale=0.6,
	every node/.style={transform shape}, 
	base/.style={circle,draw,minimum size=30pt}, 
	circ/.style={minimum size=35pt},
	triangle/.style={regular polygon, regular polygon sides = 3, draw, inner sep=0, text width=15mm}
]
\tikzstyle{level 1}=[sibling distance=60mm]
\tikzstyle{level 2}=[sibling distance=40mm]
\tikzstyle{level 3}=[sibling distance=30mm]
\node[base] (out1) {$out$}
	child { node[base] (+) {$+$} 
		child {node[base] (c3) {$3$}}
		child {node[base] (t) {$\times$}
			child{node[base] (i1) {$in_1$}}
			child{node[base] (i21) {$in_2$}}
		}
	}
	;
\end{tikzpicture}
\end{center}
\caption{Example for the numbering scheme in the proof of Theorem~\ref{thm_UFO_AC0}}
\label{fig_unifo_numbering}
\end{figure}

This numbering scheme is used for all non-input gates. 
Since input gates can have multiple successors, we number them separately: 
the $i$th input gate is numbered $(u+1, \dots, u+1, i)$ for $1 \leq i \leq n$.
The edge relation also needs to be considered separately, but we will come to how that is done precisely.

\begin{remark}
To use the numbering scheme above, we need access to the $u+1$th element of the structure over $\ttup$. 
In the case $u = n$, we do not have access to such an element. 
Therefore we actually use a vector of length $2d$ and encode each element of the aforementioned vector as two elements.
We will, however, explain the construction by length $d$ vectors, since this is just a technicality and using longer vectors would only convolute matters.
\end{remark}
In order to produce the formulae and terms for \FO-uniformity as per Definition~\ref{def_syntree}, we go through $\varphi$ similarly as we did for the proof of Theorem~\ref{thm-FOArb}.
We essentially traverse the syntax tree of $\varphi$ (as defined in Definition~\ref{def_syntree}) in a depth-first manner and define the formulae and terms iteratively.

For most of the node types in the syntax tree of $\varphi$, only a constant number of gates needs to be added, so the gate numbers can essentially be hardcoded.
For quantifiers and \sumi and \prodi constructions, the number of added gates depends on $u$. 
For those nodes, we use an existentially quantified new variable to ensure that the respective gates are correctly identified.

This idea can be seen in Figure~\ref{fig_unifo_addvar}, where for the multiplication gate numbered $(1, 1, u+1, u+1)$, we set $\varphi_{\times}(\ol{x}) \defeq \ol{x} = (1, 1, u+1, u+1)$, however, for the addition gates we need to add a quantifier.
We set $\varphi_{+}(\ol{x}) = \exists z : z \leq u \land \ol{x} = (1, 1, z, u+1)$.
An explicit application of this for a $\FO+\SUM+\PROD$-sentence is presented in Example~\ref{ex_FO=UFO_AC0}.

%

\begin{figure}
\begin{center}
\begin{tikzpicture}[
	scale=0.6,
	every node/.style={transform shape}, 
	base/.style={circle,draw,minimum size=30pt}, 
	circ/.style={minimum size=35pt},
	triangle/.style={regular polygon, regular polygon sides = 3, draw, inner sep=0, text width=15mm},
]
\tikzstyle{level 1}=[sibling distance=60mm]
\tikzstyle{level 2}=[sibling distance=40mm]
\tikzstyle{level 3}=[sibling distance=10mm]
\tikzstyle{level 4}=[sibling distance=15mm]
\node[base] (out1) {$out$}
	child { node[base] (s) {$sign$} 
		child {node[base] (t) {$\times$}
			child{node[base] (p1) {$+$}
				child{node[base] (e51) {$3$}}
				child{node[base] (e53) {$4$}}				
			}
			child{node[missing] (e41) {}}
			child{node[missing] (dots) {$\dots$}}			
			child{node[missing] (e42) {}}
			child{node[base] (p2) {$+$}
				child{node[base] (e51) {$3$}}
				child{node[base] (e53) {$4$}}				
			}
		}
	}
	;
\end{tikzpicture}
\caption{Example for specifying formulae and terms for \FO-uniformity.}
\label{fig_unifo_addvar}
\end{center}
\end{figure}

The previously described ideas are then technically executed in the following way:

Let $\varphi$ be a $\FO+\SUM+\PROD$-sentence as mentioned at the top, let the signature of $\varphi$ be $\sigma$ and let $A$ be the set of \R-structures defined by $\varphi$.
Let also $(C_n)_{n \in \N}$ be the circuit family deciding $A$ as constructed in the proof of Theorem~\ref{thm-FOArb} and let $d$ be the depth of $(C_n)_{n \in \N}$.

Before going into the explicit definitions for our formulae and terms, we need to have some auxiliary values and notations available to represent the elements we wish to use for our gate numbering.
For this purpose, we define the following helper formula, where $f_1, \dots, f_k$ are the function symbols used in $\varphi$ and $ar(f_i)$ is the arity of $f_i$ for all $i$.
\begin{align*}
\varphi_\mathrm{aux} \equiv & \mathrlap{\exists n \forall x : \rk(n) \geq \rk(x) \land} \\
					 & \mathrlap{\exists u : \sum_{1 \leq i \leq k} \rk(u)^{ar(f_i)} = \rk(n) \land} \\
					 & \mathrlap{\exists \mine \forall x : \rk(\mine) \leq \rk(x) \land} \\
					 & \exists \smine \forall x : && \rk(\smine) \neq \rk(\mine) \land \\
					 &&& (\rk(x) \neq \rk(\mine) \to \rk(\smine) \leq \rk(x))
\end{align*}
The rank of $n$ is equal to the length of the input of the circuit, and therefore the encoding length of the input structure of $\varphi$, since the formulas we define here are interpreted over structures over \ttup.
We only need this value to extract $u$, the rank of which is the size of the universe of the input structure of $\varphi$.
The values \mine and \smine represent the two lowest ranked values (in the structure, if interpreted as a structure over \ttup).

In order to define the formulae and terms necessary for \FO-uniformity, we also introduce the following notations:

We will write $u+1$ to denote the element $v$ such that $\rk(v) = \rk(u) + 1$ and $1$ and $2$ to denote min and second\_min, respectively, when there is no risk of confusing them with actual numbers. 


Let $p$ be a tuple of length $\abs{p}$ which is padded with $u+1$, i.e., $p$ is of the following form:
\[
p = (p_1, p_2, \dots, p_k, u + 1, \dots, u + 1)
\]
and let $a_1, \dots, a_l \in \{1, 2, u+1\} \cup \Vars$ for $l = \abs{p} - k$ where $\Vars$ is a set of variable symbols. 
We then write 
\[
\papp{a_1, \dots, a_l}
\]
to denote the tuple
\[
(p_1, \dots, p_k, a_1, \dots, a_l).
\]

With these notations and auxiliary values at hand, we can now start defining the formulae and terms for our \FO-uniformity.
Since we know that for every $n \in \N$, the $\rk(i)$th input gate of $C_n$ is numbered $(u+1, \dots, u+1, i)$, we set 
\[
\varphi_{\mathrm{input}}(\ol{x}, \ol{i}) \defeq \exists a : \ol{i} = (a, \dots, a)  \land \ol{x} = (u+1, \dots, u+1, a).
\]

For the remaining formulae and terms, we proceed iteratively by traversing the syntax tree of $\varphi$ in a depth-first manner.
For the purpose of clarity, in the following whenever we refer to a \emph{node}, we mean a node of the syntax tree and whenever we call something a \emph{gate}, we refer to a gate in a circuit. 
(As opposed to previously, where we used both terms to refer to gates in a circuit.)
We will eventually define $\varphi^s_{+}$, $\varphi^s_{\times}$, $\varphi^s_{\mathrm{sign}}$, $\varphi^s_{\mathrm{const}}$ and $\varphi^s_{\mathrm{const\_val}}$, where $s$ is the size of the syntax tree of $\varphi$. 
Along with a respective quantifier prefix, those formulae and terms will then represent $\varphi_{+}$, $\varphi_{\times}$, $\varphi_{\mathrm{sign}}$, $\varphi_{\mathrm{const}}$ and $\varphi_{\mathrm{const\_val}}$ as required for \FO-uniformity.
As briefly described previously, the edge relation $\varphi_\mathrm{E}$ will for the most part be defined by the non $u+1$ prefixes of two gates with the following idea:
\begin{align*}
\varphi_\mathrm{E}(\ol{x}, \ol{y}) \defeq &~ \varphi_{\mathrm{aux}} \land \varphi^s_\mathrm{quantifier\_prefix} \land \bigwedge\limits_{1 \leq i \leq d} (y_i \neq u+1 \to x_i = y_i) \land \\
& ((x_1 \neq u+1 \land x_2 = u+1 \land y_1 = u+1) \lor \\
& \bigvee\limits_{2 \leq i < d} x_i \neq u+1 \land x_{i+1} = u+1 \land y_i = u+1 \land y_{i-1} \neq u+1))
\end{align*}
However, this still leaves out the outgoing edges from input gates.
We will handle those by similarly to the gate types iteratively defining $\varphi_\mathrm{input\_edges}$ and adding them to $\varphi_\mathrm{E}$ in the endy.
We start out by setting
\begin{equation*}
\varphi^0_{+}(\ol{x}) = \varphi^0_{\times}(\ol{x}) = \varphi^0_{\mathrm{sign}}(\ol{x}) = \varphi^0_{\mathrm{const}}(\ol{x}) \defeq \bot,
\end{equation*}
\begin{equation*}
\varphi^0_{\mathrm{const\_val}}(\ol{x}) \defeq 0,
\end{equation*}
\begin{equation*}
\varphi^0_{\mathrm{quantifier\_prefix}} \defeq \top
\end{equation*}
and 
\begin{equation*}
\varphi^0_\mathrm{input\_edges}(\ol{x}, \ol{y}) \defeq \bot.
\end{equation*}
Since we know that the output gate is always numbered $(u+1, \dots, u+1)$, we also set
\[
\varphi_{\mathrm{output}}(\ol{x}) \defeq \ol{x} = (u+1, \dots, u+1).
\]
Let $\text{next\_root}$ be the gate type of the topmost gate of the circuit construction representing the root node of $\tree(\varphi)$ (as per the way its defined in the proof of Theorem~\ref{thm-FOArb}).
Then we also set 
\begin{equation}
\varphi_{\mathrm{next\_root}}^0(\ol{x}) \defeq \ol{x} = (1, u+1, \dots, u+1).
\label{eq_unifo_next_root}
\end{equation}

The remaining steps for defining our formulae and terms are taken iteratively by going through $\tree(\varphi)$ in a depth-first manner as follows:

Let $\varphi^j_{+}(\ol{x})$, $\varphi^j_{\times}(\ol{x})$, $\varphi^j_{\mathrm{sign}}(\ol{x})$, $\varphi^j_{\mathrm{const}}(\ol{x})$, $\varphi^j_{\mathrm{const\_val}}(\ol{x})$, $\varphi^j_{\mathrm{quantifier\_prefix}}(\ol{x})$ and $\varphi^j_\mathrm{input\_edges}(\ol{x}, \ol{y})$ be the functions, relations which have been defined in the previous step.

Let $p$ be the tuple representing the root gate of the circuit construction for the current node in $\tree(\varphi)$, which was previously set. 
In the first step, just after handling the output gate, $p$ is $(1, u+1, \dots, u+1)$ as set in equation~(\ref{eq_unifo_next_root}).

\begin{enumerate}
\item\label{item_exists} If the current node in $\tree(\varphi)$ is labelled $\exists y$, then set 
\begin{align*}
\varphi^{j+1}_{+}(\ol{x}) \defeq~ & \varphi^{j}_{+}(\ol{x}) \lor \ol{x} = \papp{1, u+1, \dots, u+1} 
\end{align*}
And if $\text{next\_root}$ is the gate type of the topmost gate of the circuit representing the next node in $\tree(\varphi)$ (as per the way its defined in the proof of Theorem~\ref{thm-FOArb}), then with $z \in \Vars$ being a new variable symbol we do the following:
If $\mathrm{next\_root} = +$, we add $\ol{x} = \papp{1,z, u+1, \dots, u+1}$ to the disjunction above, and otherwise we set
\[
\varphi^{j+1}_{\mathrm{next\_root}}(\ol{x}) \defeq \varphi^{j}_{\mathrm{next\_root}}(\ol{x}) \lor \ol{x} = \papp{1, z, u+1, \dots, u+1}.
\]
In both cases we also set
\[
\varphi^{j+1}_{\mathrm{quantifier\_prefix}} \defeq \exists z : \rk(z) \leq \rk(u) \land \varphi^{j}_{\mathrm{quantifier\_prefix}}.
\]
We will need the variable $z$ again when we reach variable or function nodes in the syntax tree.
Therefore we keep the index of $z$ in $\ol{x}$ in mind.
\begin{remark}
In each case all of the different $\varphi^{j+1}$ which are not explicitly mentioned are implicitly assumed to remain unchanged, i.e., $\varphi^{j+1} \defeq \varphi^{j}$.
\end{remark}
\item\label{item_forall} If the current node in $\tree(\varphi)$ is labelled $\forall y$, then proceed analogously with changing $\varphi^{j+1}_{\times}$ instead of $\varphi^{j+1}_{+}$.
\item If the current node in $\tree(\varphi)$ is labelled $\neg$, then proceed according to the translation of subtraction into our gate types as in Lemma~\ref{lem_aux_gates} as follows:
\begin{align*}
\varphi^{j+1}_{\times}(\ol{x}) \defeq~ & \varphi^{j}_{\times}(\ol{x}) \lor \ol{x} = \papp{2, u+1, \dots, u+1} \\
\varphi^{j+1}_{\mathrm{sign}}(\ol{x}) \defeq~ & \varphi^{j}_{\mathrm{sign}}(\ol{x}) \lor \ol{x} = \papp{2, 2, u+1, \dots, u+1}  \\
\varphi^{j+1}_{\mathrm{const}}(\ol{x}) \defeq~ & \varphi^{j}_{\mathrm{const}}(\ol{x}) \lor \ol{x} = \papp{1, u+1, \dots, u+1} \lor \\
& \ol{x} = \papp{2, 1, u+1, \dots, u+1}  \\
\varphi^{j+1}_{\mathrm{const\_val}}(\ol{x}) \defeq~ & \varphi^{j}_{\mathrm{const\_val}}(\ol{x}) + \\
& \mkern-18mu \chi[\varphi_\mathrm{aux} \land \varphi^j_\mathrm{quantifier\_prefix} \land \ol{x} = \papp{1, u+1, \dots, u+1}] \times 1 + \\
& \mkern-18mu \chi[\varphi_\mathrm{aux} \land \varphi^j_\mathrm{quantifier\_prefix} \land \ol{x} = \papp{2, 1, u+1, \dots, u+1}] \times -1
\end{align*}
And if $\text{next\_root}$ is the gate type of the topmost gate of the circuit representing the next node in $\tree(\varphi)$, if $\mathrm{next\_root}$ is among the gate types specified above, we add  $\ol{x} = \papp{2, 2, 1, u+1, \dots, u+1}$ to the respective disjunction.
Otherwise, we set
\begin{align*}
\varphi^{j+1}_{\mathrm{next\_root}}(\ol{x}) \defeq~&  \varphi^{j}_{\mathrm{next\_root}}(\ol{x}) \lor \ol{x} = \papp{2, 2, 1, u+1, \dots, u+1}.
\end{align*}
\begin{remark}
For $\varphi_\mathrm{const\_val}(\ol{x})$, the formula $\varphi_\mathrm{aux}$ and the quantifier prefix need to be inside the characteristic function each time, since we cannot just prepend them to a number term.
This way, we might quantify variables which are not used in this characteristic function context (for example in the case of quantification on both sides of a conjunction), but that does not cause a problem.
\end{remark}
\item If the current node in $\tree(\varphi)$ is labelled $\land$, then set
\begin{align*}
\varphi^{j+1}_{\times}(\ol{x}) \defeq~ & \varphi^{j}_{\times}(\ol{x}) \lor \ol{x} = \papp{1, u+1, \dots, u+1}
\end{align*}
And if $\text{next\_left\_root}$ and $\text{next\_right\_root}$ are the gate types of the topmost gates of the circuit representing the left and right successor node in $\tree(\varphi)$ respectively, then we also set 
\begin{align*}
\varphi^{j+1}_{\mathrm{next\_left\_root}}(\ol{x}) \defeq~&  \varphi^{j}_{\mathrm{next\_left\_root}}(\ol{x}) \lor \ol{x} = \papp{1,1, u+1, \dots, u+1} \\
\varphi^{j+1}_{\mathrm{next\_right\_root}}(\ol{x}) \defeq~&  \varphi^{j}_{\mathrm{next\_right\_root}}(\ol{x}) \lor \ol{x} = \papp{1,2, u+1, \dots, u+1}.
\end{align*}
\item If the current node in $\tree(\varphi)$ is labelled $\lor$, $\to$ or $\leftrightarrow$ then proceed analogously using the translation in Lemma~\ref{lem_aux_gates} if necessary.
\item If the current node in $\tree(\varphi)$ is labelled $=$, then proceed according to the translation in Lemma~\ref{lem_aux_gates} as in the case for $\neg$.
\item If the current node in $\tree(\varphi)$ is labelled $<$, then proceed according to the translation in Lemma~\ref{lem_aux_gates} as in the case for $\neg$.
\item If the current node in $\tree(\varphi)$ is labelled $x$ for $x \in \Vars$, then the number of the singular constant gate $g_x$ representing $x$ has already been set in the previous step.
Thus we only need to explicitly set $\varphi_{\mathrm{const\_val}}^{j+1}(\ol{x})$.
Let $nr_x$ be the value in the gate number of $g_x$ which was set when $x$ was quantified. 
Then we set 
\[
\varphi_{\mathrm{const\_val}}^{j+1}(\ol{x}) \defeq \varphi_{\mathrm{const\_val}}^{j}(\ol{x}) + \chi[\varphi_\mathrm{aux} \land \varphi^j_\mathrm{quantifier\_prefix} \land \ol{x} = p] \times \rk(nr_x).
\]
\item If the current node in $\tree(\varphi)$ is labelled $f(a_1, \dots, a_\ell)$ for a function symbol $f \in L_s$ and $a_1, \dots, a_\ell \in \Vars$, then in our circuit there is an edge from an input gate to the "dummy" addition gate $g$, of which the number was specified in the previous step.
(Here, we can also see, why a dummy gate was convenient, since otherwise the construction for $f(a_1, \dots, a_\ell)$ would not have a top gate.) 
We therefore need to specify an input edge rather than a gate number. 
Let $p_{a1}, \dots, p_{a\ell}$ be the variables in $p$ which were added, when the quantification of respective variables $a_1, \dots, a_\ell$ was specified (in the cases \ref{item_exists}, \ref{item_forall} or \ref{item_sumprod}).
We then set
\begin{align*}
\varphi^{j+1}_{\mathrm{input\_edges}}(\ol{x}, \ol{y}) \defeq~& \varphi^{j}{\mathrm{input\_edges}}(\ol{x}, \ol{y}) \lor \\
& x_1 = u+1 \land \dots \land x_{d-1} = u+1 \land \\
& \Big(\sum_{1 \leq i < k} \left( \rk(u)^{\textit{ar}(f_i)} \right) + \rk(p_{a1}) \times \rk(u)^{\ell - 1} + \dots \\
& + \rk(p_{a\ell - 1}) \times \rk(u) + \rk(p_{a\ell}) = \rk(x_d) \Big) \land \\
& \ol{y} = p,
\end{align*}
where $f_i$ is the $ith$ function in the ordering of the function symbols in $\sigma$, $f = f_k$ and $\textit{ar}(f_i)$ is the arity of $f_i$ for all $i$. 
\item If the current node in $\tree(\varphi)$ is labelled $c$ for a $c \in \R$, then as in the case for the node labelled with a variable, the gate number for the constant gate has been previously specified and we only need to ensure that the correct constant value is represented by $\varphi^{j+1}_{\mathrm{const\_val}}$.
\[
\varphi_{\mathrm{const\_val}}^{j+1}(\ol{x}) \defeq \varphi_{\mathrm{const\_val}}^{j}(\ol{x}) + \chi[\varphi_\mathrm{aux} \land \varphi^j_\mathrm{quantifier\_prefix} \land \ol{x} = p] \times c.
\]
\item If the current node in $\tree(\varphi)$ is labelled $f(x_1, \dots, x_\ell)$ for a function symbol $f \in L_f$, we proceed analogously to the case for the function symbols of $L_s$.
\item If the current node in $\tree(\varphi)$ is labelled $+$ or $\times$, then the numbering of the $+$/$\times$ gate has already been specified and we only need to make sure that the root gates for the two successor nodes are numbered correctly as well. 
If $\text{next\_left\_root}$ and $\text{next\_right\_root}$ are the gate types of the topmost gates of the circuit construction representing the left and right successor node in $\tree(\varphi)$ respectively, then we also set 
\begin{align*}
\varphi^{j+1}_{\mathrm{next\_left\_root}}(\ol{x}) \defeq~&  \varphi^{j}_{\mathrm{next\_left\_root}}(\ol{x}) \lor \ol{x} = \papp{1, u+1, \dots, u+1} \\
\varphi^{j+1}_{\mathrm{next\_right\_root}}(\ol{x}) \defeq~&  \varphi^{j}_{\mathrm{next\_right\_root}}(\ol{x}) \lor \ol{x} = \papp{2, u+1, \dots, u+1}.
\end{align*}
\item\label{item_sumprod} If the current node in $\tree(\varphi)$ is labelled $\sumi$ or $\prodi$, then the numbering of the $+$/$\times$ gate has already been specified and we again only need to make sure that the root gates of the subcircuits representing the successor node are numbered correctly. 
If $\text{next\_root}$ is the gate type of the topmost gate of the circuit construction representing the next node in $\tree(\varphi)$, then with $z \in \Vars$ being a new variable symbol we also set 
\begin{align*}
\varphi^{j+1}_{\mathrm{next\_root}}(\ol{x}) \defeq~&  \varphi^{j}_{\mathrm{next\_root}}(\ol{x}) \lor \ol{x} = \papp{z, u+1, \dots, u+1} \\
\varphi^{j+1}_{\mathrm{quantifier\_prefix}} \defeq~& \exists z : \rk(z) \leq \rk(u) \land \varphi^{j}_{\mathrm{quantifier\_prefix}}.
\end{align*}
Similarly to the case of quantifier nodes in the syntax tree, we make note of the index of $z$ in $\ol{x}$.
\end{enumerate}

\begin{remark}
For all the different $\varphi^j$, if no explicit definition for $\varphi^{j+1}$ is given, this means $\varphi^{j+1} = \varphi^{j}$.
\label{rm_unifo_phi_unchanged}
\end{remark}

We can now define the formulae and terms for \FO-uniformity as follows:
\begin{align*}
\varphi_+(\ol{x}) \defeq &~ \varphi_{\mathrm{aux}} \land \varphi^s_\mathrm{quantifier\_prefix} \land \varphi^{s}_+(\ol{x}) \\
\varphi_\times(\ol{x}) \defeq &~ \varphi_{\mathrm{aux}} \land \varphi^s_\mathrm{quantifier\_prefix} \land \varphi^{s}_\times(\ol{x}) \\
\varphi_\mathrm{sign}(\ol{x}) \defeq &~ \varphi_{\mathrm{aux}} \land \varphi^s_\mathrm{quantifier\_prefix} \land \varphi^{s}_\mathrm{sign}(\ol{x}) \\
\varphi_{\mathrm{input}}(\ol{x}, \ol{i}) \defeq &~ \varphi_{\mathrm{aux}} \land \exists a : \ol{i} = (a, \dots, a)  \land \ol{x} = (u+1, \dots, u+1, a) \\
\varphi_\mathrm{E}(\ol{x}, \ol{y}) \defeq &~ \varphi_{\mathrm{aux}} \land \varphi^s_\mathrm{quantifier\_prefix} \land \bigwedge\limits_{1 \leq i \leq d} (y_i \neq u+1 \to x_i = y_i) \land \\
& ((x_1 \neq u+1 \land x_2 = u+1 \land y_1 = u+1) \lor \\
& \bigvee\limits_{2 \leq i < d} x_i \neq u+1 \land x_{i+1} = u+1 \land y_i = u+1 \land y_{i-1} \neq u+1)) \lor \\
& \varphi^s_\mathrm{input\_edges}(\ol{x}, \ol{y})\\
\varphi_{\mathrm{output}}(\ol{x}) \defeq &~ \varphi_\mathrm{aux} \land \ol{x} = (u+1, \dots, u+1) \\
\varphi_\mathrm{const}(\ol{x}) \defeq &~ \varphi_{\mathrm{aux}} \land \varphi^s_\mathrm{quantifier\_prefix} \land \varphi^{s}_\mathrm{const}(\ol{x}) \\
\varphi_\mathrm{const\_val}(\ol{x}) \defeq &~ \varphi^{s}_\mathrm{const\_val}(\ol{x}) \\
\varphi_{\mathrm{universe}}(\ol{x}) \defeq & \varphi_+(\ol{x}) \lor \varphi_\times(\ol{x}) \lor \varphi_\mathrm{sign}(\ol{x}) \lor \exists \ol{i} : \varphi_\mathrm{input}(\ol{x}, \ol{i}) \lor \\
& \varphi_\mathrm{output}(\ol{x}) \lor \varphi_\mathrm{const}(\ol{x})
\end{align*}

Since $(C_n)_{n \in \N}$ decides $A$ and we have just shown $(C_n)_{n \in \N}$ to be \FO-uniform, we can conclude that $A \in \unifo$-$\ACO$. \smallskip

\noindent{}$\unifo\text{-}\ACO \subseteq \FO+\SUM+\PROD$:

Showing that every set decided by a \FO-uniform circuit family \C can be defined by a $\FO+\SUM+\PROD$-sentence can be show similarly to how this was shown for the non-uniform case in Theorem~\ref{thm-FOArb}.
The sentence is constructed in the same way, and all the formulae and terms used in the sentence are available to us thanks to the \FO-uniformity of \C in the following way:
\begin{align*}
t(\ol{v}) \defeq~& \chi[\exists x : \varphi_\mathrm{input}(\ol{v}, (x, \dots, x))] \times 1 + \\
& \chi[\varphi_\mathrm{const}(\ol{v})] \times 2 + \\
& \chi[\varphi_+(\ol{v})] \times 3 + \\
& \chi[\varphi_\times(\ol{v})] \times 4 + \\
& \chi[\varphi_\mathrm{sign}(\ol{v})] \times 5 + \\
& \chi[\varphi_\mathrm{output}(\ol{v})] \times 6 \\
c(\ol{v}) \defeq & \varphi_\mathrm{const\_val}(\ol{v}) \\
in(\ol{v}, i) \defeq & \chi[ \varphi_\mathrm{input}(\ol{v}, (i, \dots, i)) ] \\
pred(\ol{v}, \ol{w}) \defeq & \chi[ \varphi_\mathrm{E}(\ol{v}, \ol{w}) ]
\end{align*}

With this, we have shown that $\unifo$-$\ACO = \FO+\SUM+\PROD$. \qed
\end{proof}

\begin{example}\label{ex_FO=UFO_AC0}

The circuit $C_6$ of the \FO-uniform circuit family for the $\FO+\SUM+\PROD$-sentence $\varphi \defeq \exists x : f(x) = 3 \lor g(x) = 2$ (of which the syntax tree is depicted in Figure~\ref{fig_syntree_ex}) as constructed in the proof of Theorem~\ref{thm_UFO_AC0} can be seen in Figure~\ref{fig_unifo_ex}.

\begin{figure}[h]
\begin{center}
\begin{tikzpicture}[
	scale=0.6,
	every node/.style={transform shape}, 
	base/.style={circle,draw,minimum size=30pt}, 
	circ/.style={minimum size=35pt},
	triangle/.style={regular polygon, regular polygon sides = 3, draw, inner sep=0, text width=15mm}
]
\tikzstyle{level 1}=[sibling distance=60mm]
\tikzstyle{level 2}=[sibling distance=40mm]
\tikzstyle{level 3}=[sibling distance=50mm]
\tikzstyle{level 4}=[sibling distance=30mm]
\tikzstyle{level 5}=[sibling distance=25mm]
\tikzstyle{level 6}=[sibling distance=12mm]
\node[base] (out) {$out$}
	child { node[base] (sign1) {$\textit{sign}$} 
		child{ node[base] (p1) {$+$} 
			child{ node[base] (sign21) {$\textit{sign}$} 
				child{ node[base] (p21) {$+$}
					child{ node[base] (e11) {$=$} 
						child{ node[base] (d1) {$+$} }
						child{ node[base] (c11) {$3$} }
					}
					child{ node[base] (e12) {$=$} 
						child{ node[base] (d2) {$+$} }					
						child{ node[base] (c12) {$2$} }
					}
				}			
			}
			child{ node[base] (sign22) {$\textit{sign}$} 
				child{ node[base] (p22) {$+$}
					child{ node[base] (e21) {$=$} 
						child{ node[base] (d3) {$+$} }
						child{ node[base] (c21) {$3$} }
					}
					child{ node[base] (e22) {$=$} 
						child{ node[base] (d4) {$+$} }
						child{ node[base] (c22) {$2$} }
					}
				}			
			}
			child{ node[base] (sign23) {$\textit{sign}$} 
				child{ node[base] (p23) {$+$}
					child{ node[base] (e31) {$=$} 
						child{ node[base] (d5) {$+$} }
						child{ node[base] (c31) {$3$} }
					}
					child{ node[base] (e32) {$=$} 
						child{ node[base] (d6) {$+$} }
						child{ node[base] (c32) {$2$} }
					}
				}			
			}
		}
	}
	;
\node[base, below= 2.5cm of e11] (i1) {$in_1$}; 
\node[base, below= 2.5cm of e12] (i2) {$in_2$}; 
\node[base, below= 2.5cm of e21] (i3) {$in_3$}; 
\node[base, below= 2.5cm of e22] (i4) {$in_4$}; 
\node[base, below= 2.5cm of e31] (i5) {$in_5$}; 
\node[base, below= 2.5cm of e32] (i6) {$in_6$}; 

\draw[-] (i1) -- (d1);
\draw[-] (i2) -- (d3);
\draw[-] (i3) to[out=45, in=270] (d5);
\draw[-] (i4) to[out=135, in=270] (d2);
\draw[-] (i5) -- (d4);
\draw[-] (i6) -- (d6);


\end{tikzpicture}
\end{center}
\caption{$C_6$ of the \FO-uniform circuit family for the $\FO+\SUM+\PROD$-sentence $\varphi \defeq \exists x : f(x) = 3 \lor g(x) = 2$}
\label{fig_unifo_ex}
\end{figure}

The explicit construction of the formulae and terms for this families \FO-uniformity goes as follows:

We first define the formula $\varphi_\mathrm{aux}$ to specify our constants:
\begin{align*}
\varphi_\mathrm{aux} \equiv & \mathrlap{\exists n \forall x : \rk(n) \geq \rk(x) \land} \\
					 & \mathrlap{\exists u : \underbrace{\rk(u)}_{\rk(u)^{ar(f)}} + \underbrace{\rk(u)}_{\rk(u)^{ar(g)}} = \rk(n) \land} \\
					 & \mathrlap{\exists \mine \forall x : \rk(\mine) \leq \rk(x) \land} \\
					 & \exists \smine \forall x : && \rk(\smine) \neq \rk(\mine) \land \\
					 &&& (\rk(x) \neq \rk(\mine) \to \rk(\smine) \leq \rk(x))
\end{align*}

(We will still be using $1$, $2$ and $u+1$ to denote the variables $\mine$, $\smine$ and the one with rank $\rk(u)+1$, respectively.)

Since the circuit family has depth $7$, our tuples are also of this length.

We begin by setting
\begin{equation*}
\varphi^0_{+}(\ol{x}) = \varphi^0_{\times}(\ol{x}) = \varphi^0_{\mathrm{sign}}(\ol{x}) = \varphi^0_{\mathrm{const}}(\ol{x}) \defeq \bot,
\end{equation*}
\begin{equation*}
\varphi^0_{\mathrm{const\_val}}(\ol{x}) \defeq 0,
\end{equation*}
\begin{equation*}
\varphi^0_{\mathrm{quantifier\_prefix}} \defeq \top
\end{equation*}
and 
\begin{equation*}
\varphi^0_\mathrm{input\_edges}(\ol{x}, \ol{y}) \defeq \bot.
\end{equation*}
Now we proceed by going through the syntax tree of $\varphi$ which is shown in Figure~\ref{fig_syntree_ex}.
We fix the output number as follows
\begin{equation*}
\varphi_{\mathrm{output}}(\ol{x}) \defeq \ol{x} = (u+1, u+1, u+1, u+1, u+1, u+1, u+1),
\end{equation*}
and since the root gate of the circuit construction representing the top most $\exists x$ node of the syntax tree is a sign gate, we set
\[
\varphi_\mathrm{sign}^0(\ol{x}) \defeq \ol{x} = (1, u+1, u+1, u+1, u+1, u+1, u+1).
\]
Now the current node in $\tree(\varphi)$ is this $\exists$-node, so we proceed by setting 
\[
\varphi^1_+(\ol{x}) = \varphi^0_+(\ol{x}) \lor \ol{x} = (1, 1, u+1, u+1, u+1, u+1, u+1).
\]
Given that the next node in $\tree(\varphi)$ is a $\lor$ node and the root gate for the respective circuit construction is a $\mathrm{sign}$ gate, we set
\begin{align*}
\varphi^1_\mathrm{sign} \defeq &~ \varphi^0_\mathrm{sign}(\ol{x}) \lor \ol{x} = (1, 1, z, u+1, u+1, u+1, u+1) \\
\varphi^1_\mathrm{quantifier\_prefix} \defeq &~ \exists z : \rk(z) \leq \rk(u) \land \varphi^0_\mathrm{quantifier\_prefix},
\end{align*}
where $z$ is a new variable symbol.
As stated in Remark~\ref{rm_unifo_phi_unchanged}, we keep everything else unchanged.
Now the current node is that $\lor$ node, which is why we set
\begin{align*}
\varphi_+^2(\ol{x}) \defeq &~ \varphi^1_+(\ol{x}) \lor \ol{x} = (1, 1, z, 1, u+1, u+1, u+1).
\end{align*}
The next root is then an equality gate\footnote{\label{footnote_unifo_ex} Technically, we neither have equality gates in our circuit nor need a $\varphi_=$ formula for our uniformity. 
While we use the equality gate and formula here for brevity, we could just use the translation from Lemma~\ref{lem_aux_gates} to remove the equality gates and thus the need for the equality formula $\varphi_=$.
We believe that the idea still comes across like this and that enlarging the circuit for technical correctness would make this example needlessly convoluted.}
which is why we proceed by setting
\[
\varphi^2_=(\ol{x}) \defeq \varphi^1_=(\ol{x}) \lor \ol{x} = (1, 1, z, 1, 1, u+1, u+1).
\]
The current node is an equality node, the next left root is an addition gate and the next right root is a constant gate, thus we set
\begin{align*}
\varphi_+^3(\ol{x}) \defeq &~ \varphi_+^2(\ol{x}) \lor \ol{x} = (1, 1, z, 1, 1, 1, u+1)\\
\varphi_\mathrm{const}^3(\ol{x}) \defeq & \varphi_\mathrm{const}^2(\ol{x}) \lor \ol{x} = (1, 1, z, 1, 1, 2, u+1)
\end{align*}
Continuing on the left, we have a the function node $f(x)$ and thus need to specify an input edge as follows:
\begin{align*}
\varphi_\mathrm{input\_edges}^4(\ol{x}, \ol{y}) \defeq~& \varphi^3_\mathrm{input\_edges}(\ol{x}, \ol{y}) \lor \\
& x_1 = u+1 \land \dots \land x_{6} = u+1 \land \\
& 0 + \rk(z) \times 1 = \rk(x_7) \land \\
& \ol{y} = (1, 1, z, 1, 1, 1, u+1)
\end{align*}
Next up in our depth-first traversal of $\tree(\varphi)$ is the constant node $3$, the number of which was already specified in step $3$.
We only need to specify its value as follows:
\begin{multline}
\varphi^5_\mathrm{const\_val}(\ol{x}) \defeq \varphi^4_\mathrm{const\_val}(\ol{x}) + \\ \chi[\varphi_\mathrm{aux} \land \varphi^4_\mathrm{quantifier\_prefix} \land \ol{x} = (1, 1, z, 1, 1, 2, u+1)] \times 3
\end{multline}
Up next in our traversal is the right equality gate in $\tree(\varphi)$.
Analogously to step $2$, we set
\[
\varphi^6_=(\ol{x}) \defeq \varphi^5_=(\ol{x}) \lor \ol{x} = (1, 1, z, 1, 2, u+1, u+1).
\]
and 
\begin{align*}
\varphi_+^6(\ol{x}) \defeq &~ \varphi_+^5(\ol{x}) \lor \ol{x} = (1, 1, z, 1, 2, 1, u+1)\\
\varphi_\mathrm{const}^6(\ol{x}) \defeq &~ \varphi_\mathrm{const}^5(\ol{x}) \lor \ol{x} = (1, 1, z, 1, 2, 2, u+1).
\end{align*}
With only two nodes to go, similarly to what we did for $f(x)$, for the function symbol node $g(x)$ we also only need to specify input edges:
\begin{align*}
\varphi_\mathrm{input\_edges}^7(\ol{x}, \ol{y}) \defeq~& \varphi^6_\mathrm{input\_edges}(\ol{x}, \ol{y}) \lor \\
& x_1 = u+1 \land \dots \land x_{6} = u+1 \land \\
& \underbrace{\rk(u)}_{rk(u)^{ar(f)}} + \rk(z) \times 1 = \rk(x_7) \land \\
& \ol{y} = (1, 1, z, 1, 2, 1, u+1)
\end{align*}
For the last node in our traversal of $\tree(\varphi)$, the constant node $2$, we proceed similarly to how we did for $3$ in step $5$:
\begin{multline}
\varphi^8_\mathrm{const\_val}(\ol{x}) \defeq \varphi^7_\mathrm{const\_val}(\ol{x}) + \\ \chi[\varphi_\mathrm{aux} \land \varphi^7_\mathrm{quantifier\_prefix} \land \ol{x} = (1, 1, z, 1, 2, 2, u+1)] \times 2
\end{multline}

Now taken all together, we can explicitly write our uniformity formulae and terms out as follows:
%
%

\begin{align*}
\varphi_+(\ol{x}) \defeq~& \varphi_\mathrm{aux} \land \exists z : \bot \lor \\
& \ol{x} = (1, 1, u+1, u+1, u+1, u+1, u+1) \lor \\
& \ol{x} = (1, 1, z, 1, 1, 1, u+1) \lor \\
& \ol{x} = (1, 1, z, 1, 2, 1, u+1) \\
\varphi_\times(\ol{x}) \defeq~& \varphi_\mathrm{aux} \exists z : \bot \\
\varphi_\mathrm{sign}(\ol{x}) \defeq~& \varphi_\mathrm{aux} \land \exists z : \bot \lor \\
& \ol{x} = (1, u+1, u+1, u+1, u+1, u+1, u+1) \lor \\
& \ol{x} = (1, 1, z, u+1, u+1, u+1, u+1) \\
\varphi_{\mathrm{input}}(\ol{x}, \ol{i}) \defeq~& \varphi_{\mathrm{aux}} \land \exists a : \ol{i} = (a, a, a, a, a, a, a)  \land \\
& \ol{x} = (u+1, u+1, u+1, u+1, u+1, u+1, a) \\
\end{align*}
\begin{align*}
\varphi_\mathrm{E}(\ol{x}, \ol{y}) \defeq &~ \varphi_{\mathrm{aux}} \land \exists z: ~ \bigwedge\limits_{1 \leq i \leq 7} (y_i \neq u+1 \to x_i = y_i) \land \\
& ((x_1 \neq u+1 \land x_2 = u+1 \land y_1 = u+1) \lor \\
& \bigvee\limits_{2 \leq i < 7} x_i \neq u+1 \land x_{i+1} = u+1 \land y_i = u+1 \land y_{i-1} \neq u+1)) \lor \\
& \bot \lor \\
& x_1 = u+1 \land \dots \land x_{6} = u+1 \land \\
& 0 + \rk(z) \times 1 = \rk(x_7) \land \\
& \ol{y} = (1, 1, z, 1, 1, 1, u+1) \lor \\
& x_1 = u+1 \land \dots \land x_{6} = u+1 \land \\
& \rk(u) + \rk(z) \times 1 = \rk(x_7) \land \\
& \ol{y} = (1, 1, z, 1, 2, 1, u+1) \\
\varphi_{\mathrm{output}}(\ol{x}) \defeq &~ \varphi_\mathrm{aux} \land \ol{x} = (u+1, u+1, u+1, u+1, u+1, u+1, u+1) \\
\varphi_\mathrm{const}(\ol{x}) \defeq &~ \varphi_\mathrm{aux} \land \exists z : \bot \lor \\
&\ol{x} = (1, 1, z, 1, 1, 2, u+1) \lor \\
&\ol{x} = (1, 1, z, 1, 2, 2, u+1) \\
\varphi_\mathrm{const\_val}(\ol{x}) \defeq &~ 0 \\
& + \chi[\varphi_\mathrm{aux} \land \exists z :~ \ol{x} = (1, 1, z, 1, 1, 2, u+1)] \times 3 \\
& + \chi[\varphi_\mathrm{aux} \land \exists z :~ \ol{x} = (1, 1, z, 1, 2, 2, u+1)] \times 2 \\
\varphi_{\mathrm{universe}}(\ol{x}) \defeq & \varphi_+(\ol{x}) \lor \varphi_\times(\ol{x}) \lor \varphi_\mathrm{sign}(\ol{x}) \lor \exists \ol{i} : \varphi_\mathrm{input}(\ol{x}, \ol{i}) \lor \\ 
& \varphi_\mathrm{output}(\ol{x}) \lor \varphi_\mathrm{const}(\ol{x})
\end{align*}
and technically 
\begin{align*}
\varphi_=(\ol{x}) \defeq~& \varphi_\mathrm{aux} \land \exists z : \bot \lor \\
& \ol{x} = (1, 1, z, 1, 1, u+1, u+1)
& \ol{x} = (1, 1, z, 1, 2, u+1, u+1),
\end{align*}
though as discussed in the footnote on page~\pageref{footnote_unifo_ex}, those gates would ordinarily have to be tackled according to the translation in Lemma~\ref{lem_aux_gates}.
\end{example}

\begin{remark}
Even though we have only considered functional \R-structures in this paper, our findings can be generalized for \R-structures which use relations as well, since any relation can be expressed via its characteristic function.
\end{remark}

%

\section{Conclusion}

We showed that the computational power of circuits of polynomial size and constant depth over the reals can be characterized in a logical way by first-order logic on metafinite structures. 
This result is in analogy to corresponding characterizations for Boolean circuits \cite{DBLP:journals/siamcomp/Immerman89} and arithmetic circuits \cite{DBLP:journals/apal/HaakV19}.
The results presented in this paper mostly do not make use of any special properties of the real numbers and can 
be generalized for other fields with suitably adapted logic and circuit definitions.

In the Boolean and arithmetic context, it is known \cite{DBLP:journals/jcss/BarringtonIS90} that the numerical predicates of addition and multiplication play a special role: If we enhance first-order logic by these, we obtain a logic as powerful as dlogtime-uniform $\mathrm{AC}^0$-circuits, i.e., U$_\text{D}\text{-}\mathrm{AC}^0=\mathrm{FO}[+,\times]$ (see also \cite{DBLP:books/daglib/0097931}).
This does not seem to hold in our case of computation over the real numbers: $\unil\text{-}\ACO$ looks more powerful than $\FO[+,\times]$, since real numbers can be manipulated more generally by \R-machines operating in logarithmic time than in first-order formulas, because it seems that a logarithmic number of operations on reals cannot be simulated in first-order logic.
Maybe an equivalence can be obtained with a more powerful logic for real numbers, but this is a question for further research. 
However, an analogue to the Boolean equality holds if we consider uniformity defined itself in a logical way: The identity U$_\text{FO}\text{-}\mathrm{AC}^0=\mathrm{FO}[+,\times]$, well known in the Boolean world, holds in the real setting as well.


While investigating uniform circuit classes over the reals, we found that uniformity behaves somewhat differently in the real setting than it does in the Boolean one. In the classical setting, the question of uniformity arises quite naturally, since small classes like non-uniform $\mathrm{AC}^0$ already contain undecidable problems with respect to Turing machines. In the case of \ACO{} and \R-machines, the same is at least not quite obvious and worth looking into further.

We consider it worthwhile to study logical characterizations of analogues of further circuit classes of unbounded or semi-unbounded fan-in, most prominently $\mathrm{SAC}^1_\R$ and  $\mathrm{AC}^1_\R$. In the theory of arithmetic complexity, i.e., computation over arbitrary semi-rings, first an analogue of Immerman's Theorem was shown in \cite{DBLP:journals/apal/HaakV19}, and this was later used to obtain logical characterizations of the larger arithmetic classes $\mathrm{\#NC}^1$, $\mathrm{\#SAC}^1$ and  $\mathrm{\#AC}^1$ \cite{DBLP:conf/lics/0001HV18}. Remarkably these characterizations did not build on logics with repeated quantifier blocks (like in \cite{DBLP:journals/siamcomp/Immerman89}) or restricted fixed-point logic (like in \cite{DBLP:journals/jsyml/CuckerM99}). Instead, new logical characterizations of the Boolean classes $\mathrm{NC}^1$, $\mathrm{SAC}^1$ and  $\mathrm{AC}^1$ were given, somewhat similar to earlier ideas from Compton and Laflamme \cite{DBLP:journals/iandc/ComptonL90}, and these were then shifted to the arithmetic setting. Maybe this can also be useful in our context to develop characterizations for $\mathrm{SAC}^1_\R$ and  $\mathrm{AC}^1_\R$ (and maybe obtain a new characterization of $\mathrm{AC}^1_\R$).

In the theory of computation over the reals, separations among classes are known which are widely open in the discrete world; we only mention the separation of $\mathrm{NC}_\R$ and $\mathrm{P}_\R$ \cite{DBLP:journals/jc/Cucker92}. In the circuit world, the most prominent open question is if $\mathrm{TC}^0=\mathrm{NC}^1$ (see the discussion in \cite{DBLP:books/daglib/0097931}). In our context, it is intriguing to study the landscape between $\ACO$ and $\mathrm{NC}^1_\R$. Is there any meaningful way to add computational power to $\ACO$ without already arriving at the full power of $\mathrm{NC}^1_\R$? Observe that up to date,  no reasonable real analogue of the class $\mathrm{TC}^0$ is known. In Boolean complexity, $\mathrm{TC}^0$ is obtained by enriching $\mathrm{AC}^0$-circuits with majority gates. Here, the class $\ACO$ is closed under all reasonable forms of majority and threshold operations. A first step forward will be to separate $\ACO$ and $\mathrm{NC}^1_\R$, a real world analogue of a classical circuit separation from the eighties \cite{DBLP:journals/mst/FurstSS84}.

%
%
%
\bibliographystyle{splncs04}
%
\bibliography{RealAC0} 
\newpage

\appendix
\end{document}